\documentclass[pdflatex,sn-arxiv,Numbered]{sn-jnl}
\usepackage{orcidlink} 

\usepackage{amsmath}
\usepackage{amssymb}
\usepackage{amsthm}
\usepackage{manyfoot}
\usepackage{xcolor}
\usepackage[algoruled,lined,algonl,procnumbered]{algorithm2e}
\usepackage{tikz}
\usepackage{xspace}
\usepackage{subcaption}
\usepackage{enumitem}
\usepackage{tabularx}

\usetikzlibrary{calc}
\usetikzlibrary{arrows.meta,shapes,decorations.pathreplacing,decorations.pathmorphing,decorations.markings,plotmarks,patterns,patterns.meta}

\newcommand{\triangularGrid}[4]{
    \newcommand\rowNum{#1} % Half the height of grid
    \newcommand\colNum{#2} % Number of triangles of base -1 
    \newcommand\triSide{#3} % The length of each side of the cells
    \newcommand\triHeight{#4} % The height of the cells

    \clip (-.5*\triSide,-1) rectangle (\triSide+\triSide*\colNum -.5*\triSide,2.35*\triHeight*\rowNum) ;

    \foreach \i in {0,...,\rowNum} {
        \ifnum\i=0 
            \foreach \j in {0,...,\colNum} {
                \draw[gray!50] (0:\j*\triSide) -- ++ (-120:\triSide) -- ++ (0:\triSide) --++ (120:\triSide) --++ (0:\triSide) ;
            }
            \foreach \j in {0,...,\colNum} {
                \draw[gray!50] (0:\j*\triSide-\triSide) --++ (0:\triSide) --++ (120:\triSide) --++ (0:\triSide) --++ (-120:\triSide);
            }
        \else
            \foreach \j in {0,...,\colNum} {
                \newcommand\startingPoint{2*\i*\triHeight}
                \draw[gray!50] (90:\startingPoint) ++ (0:\j*\triSide) -- ++ (-120:\triSide) -- ++ (0:\triSide) --++ (120:\triSide) --++ (0:\triSide) ;
            }
            \foreach \j in {0,...,\colNum} {
                \newcommand\startingPoint{2*\i*\triHeight}
                \draw[gray!50] (90:\startingPoint) ++ (0:\j*\triSide-\triSide) --++ (0:\triSide) --++ (120:\triSide) --++ (0:\triSide) --++ (-120:\triSide);
            }
        \fi
    }
}

\theoremstyle{thmstyleone}%
\newtheorem{theorem}{Theorem}
\newtheorem{case}{Case}
\newtheorem{subcase}{Case}[case]
\newtheorem{observation}{Observation}
\newtheorem{remark}{Remark}
\newtheorem{claim}{Claim}
\newtheorem{lemma}{Lemma}

\newcommand{\msa}{\ensuremath{\mathsf{a}\xspace}}
\newcommand{\msb}{\ensuremath{\mathsf{b}\xspace}}
\newcommand{\msp}{\ensuremath{\mathsf{p}\xspace}}
\newcommand{\msq}{\ensuremath{\mathsf{q}\xspace}}
\newcommand{\msr}{\ensuremath{\mathsf{r}\xspace}}
\newcommand{\msx}{\ensuremath{\mathsf{x}\xspace}}

\newcommand{\cA}{\ensuremath{\mathcal{A}\xspace}}
\newcommand{\cP}{\ensuremath{\mathcal{P}\xspace}}
\newcommand{\cS}{\ensuremath{\mathcal{S}\xspace}}

\newcommand{\view}{\texttt{view}\xspace}

\begin{document}

\title[Self-Stabilising LE with Constant Memory]{Deterministic Self-Stabilising Leader Election for Programmable Matter with Constant Memory\footnote{This is an extended version of the work presented in \cite{chalopin2024stabilising}}}

\author*[1]{\fnm{Jérémie} \sur{Chalopin} \orcidlink{0000-0002-2988-8969}} \email{jeremie.chalopin@lis-lab.fr}

\author*[1]{\fnm{Shantanu} \sur{Das} \orcidlink{0000-0003-4008-2445}} \email{shantanu.das@lis-lab.fr}

\author*[1]{\fnm{Maria} \sur{Kokkou} \orcidlink{0009-0009-8892-3494}} \email{maria.kokkou@lis-lab.fr}

\affil[1]{\orgname{Aix Marseille Univ, CNRS, LIS}, \orgaddress{\street{163 Av. de Luminy}, \city{Marseille}, \postcode{13009}, \country{France}}}

\abstract{The problem of electing a unique leader is central to all distributed systems, including programmable matter systems where particles have constant size memory. In this paper, we present a silent self-stabilising, deterministic, stationary, election algorithm for particles having constant memory, assuming that the system is simply connected. Our algorithm is elegant and simple, and requires constant memory per particle. We prove that our algorithm always stabilises to a configuration with a unique leader, under a  daemon satisfying some fairness guarantees ({\it Gouda fairness} \cite{Gouda2001theory}). We use the special geometric properties of programmable matter in 2D triangular grids to obtain the first self-stabilising algorithm for such systems. This result is surprising since it is known that silent self-stabilising algorithms for election in general distributed networks require $\Omega(\log{n})$ bits of memory per node, even for ring topologies \cite{DolevGS99}.}

\keywords{Leader Election, Programmable Matter, Self-Stabilisation, Simply Connected, Gouda fair Daemon, Constant Memory} 

\maketitle

\bmhead{Declarations}
The authors have no competing interests to declare that are relevant to the content of this article. All authors contributed equally to all aspects of this work. All authors read and approved the final manuscript.

\bmhead{Acknowledgements}

This work has been partially supported by the National Research Agency (Agence Nationale de la Recherche (ANR)) project DUCAT (ANR-20-CE48-0006)

\section{Introduction}

Leader election (introduced by Le Lann~\cite{LeLann77}) allows to distinguish a unique process in the system as a \textit{leader}. The leader process can then act as an initiator or a coordinator, for solving other distributed problems. Thus, election algorithms are often used as building blocks for many problems in this domain. We are interested in deterministic election algorithms that are {\em self-stabilising}.
Since the seminal work of Dijkstra~\cite{Dijkstra74}, the
self-stabilisation paradigm has been thoroughly
investigated (see~\cite{DolevBook} for a survey).  A distributed algorithm is
self-stabilising if when executed on a distributed system in an
arbitrary global initial configuration, the system eventually reaches
a legitimate configuration. Self-stabilising protocols are able to
autonomously recover from transient memory failures, without external
intervention. A self-stabilising algorithm is \emph{silent} if the system always 
reaches a configuration where the processes no longer change their
states. In the self-stabilising setting, leader election is
particularly important, as many self-stabilising algorithms rely on
the existence of a distinguished node.

The concept of silent self-stabilising algorithms is also related to
{\em proof-labelling} schemes~\cite{KormanKP10} where each node is
given a local certificate to verify certain global properties of the
system (e.g., the existence of a unique leader). Each node can
check its own certificate and those of its neighbours to verify it
is in a correct configuration. If the global configuration is
incorrect, at least one node should be able to detect an inconsistency
using the local certificates. In this case, this node will change its
state, leading its neighbours to change their states and so on, until
the system stabilises to a correct configuration. 
%Blin et al.~\cite{BlinFP14} proved that from any proof-labelling scheme where each process has a certificate of size $\ell$, one can build a silent self-stabilising algorithm using $O(\ell + \log n)$ bits of memory per process, where $n$ is the number of processes in the network. This result is particularly relevant for designing self-stabilising protocols for classic problems such as leader election. A common approach in such protocols is to construct a spanning tree rooted at the leader, with each node pointing to its parent in the tree. In order to detect cycles when the system is in an incorrect state, the local certificate at each node additionally includes the hop-distance to the root. From these certificates a self--stabilising algorithm can be easily constructed using \cite{BlinFP14}, however the memory requirement per node for such an algorithm would depend on the size of the system.
Blin et
al.~\cite{BlinFP14} proved that from any proof-labelling scheme where
each process has a certificate of size $\ell$, one can build a silent
self-stabilising algorithm using $O(\ell + \log n)$ bits of memory per
process, where $n$ is the number of processes in the network. This result is particularly relevant for deriving self-stabilising protocols for classic problems, such as leader election, after constructing a certificate. A common approach in constructing a certificate for leader election is to build a spanning tree rooted at the leader, with all other nodes
pointing towards their parent in the tree. In order to detect cycles
when the system is in an incorrect state, the local certificate at
each node includes the hop-distance to the root, in addition to the
pointer to the parent. So in this approach the size of the certificate
depends on the size of the system.

% Blin et
% al.~\cite{BlinFP14} proved that from any proof-labelling scheme where
% each process has a certificate of size $\ell$, one can build a silent
% self-stabilising algorithm using $O(\ell + \log n)$ bits of memory per
% process, where $n$ is the number of processes in the network. One of
% the standard techniques for self-stabilising leader election is to
% build a spanning tree rooted at the leader, with all other nodes
% pointing towards their parent in the tree. In order to detect cycles
% when the system is in an incorrect state, the local certificate at
% each node includes the hop-distance to the root, in addition to the
% pointer to the parent. So the size of the certificate
% depends on the size of the system.

In this paper, we consider {\em programmable matter} systems which are distributed systems consisting of small, intelligent particles that connect to other particles to construct a given shape and can autonomously change shapes according to input signals. Such systems should be scalable to arbitrary sizes, so the particles are assumed to have constant size memory independent of the size of the system, similar to finite state automata. This requirement also implies that the particles are anonymous (i.e., do not have unique identifiers) and all communication is limited to $O(1)$ size messages. One well-studied model for programmable matter is the {\em Amoebot} model~\cite{AM-derakhshandeh2014amoebot} where particles operate on a triangular grid (see Section \ref{sec:rel-work} for details). Leader election is a well studied problem in this model. Informally, a particle system is said to be \emph{simply connected} when it is connected and does not contain \emph{holes} (i.e., sets of nodes that are not occupied by particles and are surrounded by particles). We refer to the closed curve of smallest area in the plane that encloses all particles in the system as the \emph{boundary}. When the system is simply connected, there are stationary deterministic algorithms for election based on the {\em erosion} approach~\cite{diluna2020shape} where the algorithm starts by deactivating particles on the boundary and moving inwards, until the last active node becomes the leader. This approach works under the minimum assumptions on the system and is the inspiration for our work.

Tolerating faults is important for programmable matter, however none of the existing algorithms for election in these systems are self-stabilising. The question is: given the constant memory of particles, is it still possible to obtain a self-stabilising algorithm for programmable matter, using other properties of such systems? We answer this question in the affirmative, at least for \emph{simply connected} programmable matter systems, showing that in this case, a deterministic silent self-stabilising algorithm for leader election is indeed possible. We use the property of such systems that there is a unique boundary in the system that is well defined, such that any particle can determine whether it is on the boundary.

\subsection{Our results}
We present a silent self-stabilising, deterministic, stationary,
election algorithm for constant-memory particles, in a simply connected system. We prove our algorithm always
stabilises to a unique leader configuration, under a
scheduler with some fairness guarantees.

%% Our technique

We first present a constant--memory proof labelling scheme ensuring the existence of a
unique leader. Our certificate orients the edges of
the network and a configuration is valid when: every edge is oriented, every particle has at most three outgoing edges, outgoing edges appear consecutively around each particle and there are no directed triangles. Note that our certificate does not
ensure that the global orientation of the network is acyclic. However,
using the geometric properties of the configuration, we are able to
show that any valid configuration has a unique sink.
As we are interested in a constant memory algorithm, one cannot transform our proof labelling scheme into
a self-stabilising algorithm using~\cite{BlinFP14}. However, we design a very simple
algorithm to orient the edges of the system. Although we do not transform the certificate generically into a self-stabilising algorithm, we do use it in an ad-hoc way within our algorithm. We show that under our
fairness assumption, one always reaches a valid configuration and that
this configuration contains a unique sink that is designated as the
leader.  

%% On the strong fairness
Following the classification of~\cite{dubois2011taxonomy}, our
scheduler is Gouda fair~\cite{Gouda2001theory}: for any configuration
$C$ that appears infinitely often in the execution, any successor
$C'$ of $C$ also appears infinitely often. Since we consider finite systems of particles with finite memory, the number of possible configurations is finite. %As a result, the number of connected configurations with a fixed number of particles is also finite. 
In this setting, Gouda fairness ensures that any configuration that is infinitely often reachable is eventually reached. Observe that a scheduler that at each step
activates a particle chosen uniformly at random is a
Gouda fair sequential scheduler.

%% On the simple connectivity
We do not assume that there exists an agreement on the orientation of
the grid, or even on its chirality. Observe that without simple connectivity and without agreement on
orientation or chirality, it is possible to construct arbitrarily
large rings of even size where all processes have the same geometric
information about the system (see Fig.~\ref{fig:not-SC}). In this setting, the results of~\cite{DolevGS99} show that there is no silent
self-stabilising leader election algorithm using constant memory. This is part of our motivation for considering simply connected systems.

\begin{figure}[htb]
\centering
\includegraphics[scale=.8,page=1]{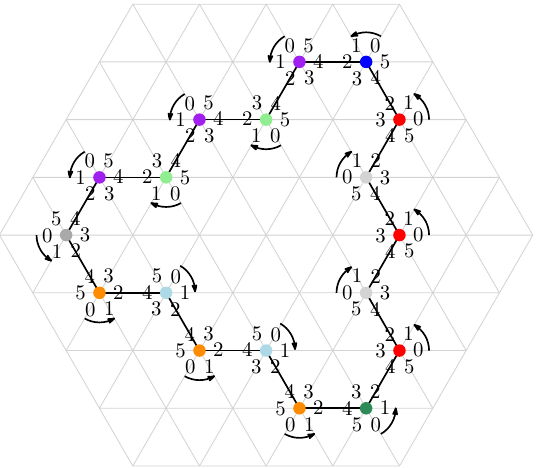}
\caption{An 18-particle ring where for each particle, the occupied
  neighbours are reached through port numbers 2 and 4. Two nodes have
  the same colour if they agree on the grid orientation.}
\label{fig:not-SC}
\end{figure}

\subsection{Related Work} \label{sec:rel-work}
In general networks, there is no self-stabilising leader
election algorithm where each process has a constant memory. More
precisely, Dolev et al.~\cite{DolevGS99} established that any silent
self-stabilising algorithm electing a leader in the class of rings
requires $\Omega(\log n)$ bits of memory per process (where $n$ is the
size of the ring). This lower bound only uses the assumption that
there exists a silent correct configuration and holds for any kind of
scheduler. More recently, Blin et al.~\cite{BlinFB23} showed that
non-silent self-stabilising algorithms require $\Omega(\log \log n)$
bits of memory per process in order to elect a leader synchronously in
the class of rings. Note that these lower bounds are tight in the
sense that for rings, there exist silent (resp. non-silent)
self-stabilising leader election algorithms using $O(\log n)$ (resp.,
$O(\log \log n)$) bits of memory per
process~\cite{DattaLV11,BlinT18}. Constant memory self stabilising algorithms for rings can be designed under special assumptions,
as in \cite{ItkisLS95} which gives an algorithm for prime sized rings
assuming a sequential scheduler. However, \cite{DolevGS99} established that this algorithm cannot be made silent.

There exists a large literature about distributed systems where each
process has finite memory. Cellular automata, introduced in the 40s in \cite{VonNeumann} are one of the best
known models of this kind. More recently, numerous papers have been
devoted to population protocols introduced in \cite{AngluinAER07}. In this model, there is a population of
finite-state agents and at each step, a scheduler picks two agents
that jointly update their states according to their current
states. The scheduler satisfies the same fairness condition as the one
we consider in this paper: any configuration that is infinitely often
reachable is eventually reached.  In this setting, there exist
election protocols using only two states when all agents start in the
same state. However, when considering self-stabilising leader election
in this setting, Cai et al.~\cite{CaiIW12} showed that a protocol using
$n-1$ states cannot solve the problem in a population of $n$ agents. This
shows that even with a Gouda fair scheduler, it is not always possible
to solve the leader election problem in a self-stabilising way when
processes have constant memory. More recently, most of the work in population protocols has focused on the complexity aspects of the protocols assuming a randomised scheduler. Since a randomised scheduler is Gouda fair, there do not exist algorithms that work under a Gouda fair scheduler but not under a randomised one. To the best of our knowledge, there additionally do not exist any computability results that can be achieved under a randomised scheduler but not under a Gouda fair one.

\begin{table*}[t]
  \centering
  \renewcommand{\arraystretch}{1.5}  % Adds vertical padding
  
  \resizebox{\textwidth}{!}{%
  \begin{tabular}{>{\centering\arraybackslash}p{3cm}|
                  >{\centering\arraybackslash}p{2cm}|
                  >{\centering\arraybackslash}p{3cm}|
                  >{\centering\arraybackslash}p{2cm}|
                  >{\centering\arraybackslash}p{2cm}|
                  >{\centering\arraybackslash}p{2cm}}
      \toprule
      \textbf{Paper} & 
      \textbf{Leaders} & 
      \textbf{Simply Connected} & 
      \textbf{Chirality} & 
      \textbf{Movement} & 
      \textbf{Seq.\ Scheduler} \\
      \midrule
      \cite{gastineau2018distributed} & 1 & \checkmark & \checkmark & \texttt{X} & \checkmark \\ 
      \hline
      \cite{emek2019deterministic} & 1 & \texttt{X} & \texttt{X} & \checkmark & \checkmark \\
      \hline
      \cite{dufoulon2021efficient} & 1 & \texttt{X} & \checkmark & \checkmark & \checkmark \\
      \hline
      \cite{dufoulon2021efficient} & 6 & \texttt{X} & \checkmark & \texttt{X} & \checkmark \\
      \hline
      \cite{diluna2020shape} & 3 & \checkmark & \texttt{X} & \texttt{X} & \texttt{X} \\ 
      \hline
      \cite{bazzi2019stationary} & 6 & \texttt{X} & \checkmark & \texttt{X} & \texttt{X} \\  
      \hline
      \cite{briones2023asynchronous} & 1 & \checkmark & \texttt{X} & \texttt{X} & \checkmark \\
      \hline
      \cite{chalopin2024deterministic} & 1 & \texttt{X} & \texttt{ X}$^*$ & \texttt{X} & \texttt{X} \\
      \toprule
  \end{tabular}
  }
  \caption{Deterministic leader election in regular triangular grids in the non self--stabilising setting. ``Simply Connected'' refers to the particle system not having \textit{Holes}. ``Chirality'' is a common sense of rotational orientation. ``Movement'' is the ability of particles to move to neighbouring nodes. A ``Sequential Scheduler'' activates one particle at a time. 
  $^{*}$Uses agreement on the directions along one axis.
  }
  \label{tab:previous-LE}
\end{table*}

Programmable matter was introduced in \cite{toffoli1991programmable} and has since gained popularity. Several models have been introduced, such as \cite{hawkes2010programmable,woods2013active,fekete2021cadbots}. In this paper we consider the well studied \textit{Amoebot} model \cite{AM-derakhshandeh2014amoebot,daymude2023canonical}. In this model, constant-memory computational entities, called \textit{particles}, operate in a triangular grid. Each node of the grid is occupied by at most one particle and particles can determine whether nodes at distance one are occupied by particles. Each particle can communicate with its neighbours by reading their respective registers. The problem of Leader Election has been studied in the specific context of programmable matter in both two dimensional (e.g., \cite{bazzi2019stationary, diluna2020shape,dufoulon2021efficient,emek2019deterministic,gastineau2018distributed,chalopin2024deterministic}) and three dimensional settings (e.g., \cite{gastineau2022leader, briones2023asynchronous}) and both deterministic (e.g., Table \ref{tab:previous-LE}) and randomised algorithms (e.g., \cite{derakhshandeh2015leader}) have been proposed. It is usually assumed that particles do not have any global sense of direction, while some papers (unlike this paper) assume that the particles have a common sense of rotational orientation, called {\em chirality} (e.g., \cite{gastineau2018distributed,dufoulon2021efficient,bazzi2019stationary}) or that particles agree on a common
direction (e.g., \cite{chalopin2024deterministic}). Particles in the Amoebot model have the ability to move to neighbouring nodes (e.g., \cite{dufoulon2021efficient,emek2019deterministic}), which we do not use here.  All known leader election algorithms operate under a fair scheduler (i.e., every particle is activated infinitely often). Although not stated explicitly, we believe that most algorithms, especially the erosion based ones such as \cite{dufoulon2021efficient,diluna2020shape}, also work under an unfair scheduler (i.e., at any time at least one particle that can update its state is activated). Although most known algorithms assume sequential particle activations (e.g., \cite{gastineau2018distributed,dufoulon2021efficient,emek2019deterministic,briones2023asynchronous}), \cite{chalopin2024deterministic,bazzi2019stationary,diluna2020shape} assume particles are activated asynchronously. The existing algorithms for leader election in programmable matter can be categorised based on the use of two main techniques: \textit{erosion} (e.g., \cite{diluna2020shape,gastineau2018distributed}) and message passing on boundaries (e.g., \cite{bazzi2019stationary,derakhshandeh2015leader}). A summary of results on deterministic leader election in the non self--stabilising case is given in Table \ref{tab:previous-LE}.

Research on self-stabilisation in the programmable matter setting is more limited. In \cite{derakhshandeh2015leader}, a randomised leader election algorithm is given and the authors discuss the possibility of making it self-stabilising by combining it with techniques from \cite{awerbuch1994memory,itkis1994fast}. However, it is assumed that particles have $O(\log^*n)$ memory. In the same paper, it is also argued that self-stabilisation in programmable matter is not possible for problems where movement is needed, as the system can become permanently disconnected. More recently, in \cite{daymude2021bio} a self-stabilising algorithm for constructing a spanning forest was introduced. In this case, the algorithm is deterministic and particles have constant memory. However, it is assumed that at least one non-faulty special particle always remains in the system. The need to extend the Amoebot model to also address self-stabilising algorithms is also discussed in \cite{daymude2023canonical}.   

\section{Model}

Let $G_\Delta$ be an infinite regular triangular grid where each node
has six neighbours. A connected particle system, $\mathcal{P}$, is simply connected if $G_\Delta \backslash \mathcal{P}$ is connected. We assume each node of the simply connected $\mathcal{P}$
contains exactly one particle. We call nodes that are in $\mathcal{P}$
\textit{occupied} and those that are not in $\mathcal{P}$,
\textit{empty}. Each particle is anonymous, has constant memory
and is stationary (i.e., does not move). A particle is
incident to six {\em ports}, leading to consecutive neighbouring nodes
in $G_\Delta$. Each port is associated with a label so that ports $i$
and $i+1 \mod 6$ lead to neighbouring nodes.
A particle knows if each port leads to an occupied or
empty node. For each occupied neighbour $q$, particle $p$ knows the label assigned by $q$ to the edge $qp$. This also leads to the following observation.

\begin{observation}\label{obs:consec-ports}
  If $pqr$ is a three--particle triangle, then the ports connecting $r$ to $p$ and
  $q$ are consecutive.
\end{observation}

Each particle has a constant-size register with arbitrary initial contents. A particle can read the register of each occupied neighbour but can only write in its own register. All particles are {\em
inactive} unless activated by the scheduler. An activated particle reads the contents of its register and
the register of each of the neighbouring particles.
Based on this information it updates the contents of its own register
according to the given algorithm.

We call $\mathcal{P}$ the \textit{support} of the particle
system. The configuration $C$ of the system at any time, consists of
the set $\mathcal{P}$ and the contents of the registers of each
particle in $\mathcal{P}$. A distributed algorithm $\cA$ is a set of local rules that particles execute. The rules of the
algorithm depend only on the content of the registers of the particles
and of its neighbours and they modify only the register of the
particle. For an algorithm $\cA$, a configuration $C$, and a particle
$p$, we say that $p$ is \textit{activable} in $C$, if the execution of
$\mathcal{A}$ modifies the contents of the register of $p$. For two
configurations $C$ and $C'$ that have the same support, we say that
$C'$ is a successor (resp. sequential successor) of $C$ if there exists a non-empty set of activable particles $Q$ (resp. an activable particle $p$)
in $C$ such that, when all particles in $Q$ execute (resp. when $p$ executes) $\mathcal{A}$, $C'$ is
obtained. An execution $\mathcal{S}$ is an infinite
sequence of configurations $\mathcal{S} = C_0, C_1, \ldots$ such that
for any $i$, $C_i$ and $C_{i+1}$ have the same support and either
there exists a non-empty set $Q_i$ of activable particles such that when all particles from $Q_i$ execute
$\mathcal{A}$ concurrently in $C_i$, $C_{i+1}$ is obtained, or there is no
activable particle and $C_{i+1} = C_i$. If $|Q_i| = 1$ for any $i$, then we say that $\mathcal{S}$ is a \textit{sequential} execution. If there exists a step where
$C_{i+1} = C_i$, we call $C_i$ a \textit{final configuration}. An
execution is \textit{Gouda fair}
\cite{dubois2011taxonomy,Gouda2001theory} if for any configuration $C$
that appears infinitely often in the execution, any successor $C'$
of $C$ also appears infinitely often. A sequential execution is \textit{sequential Gouda fair} if for any configuration $C$ that appears infinitely often in the execution, any successor $C'$ of $C$ that is obtained by activating one particle of $C$ also appears infinitely often. The notion of \textit{valid} configurations depends on the
algorithm. In the next section, we define the valid configurations we
consider in this paper. An algorithm $\mathcal{A}$ is
\textit{silent} if the set of valid configurations is precisely the set of final configurations. 

\begin{remark} \label{rem:sequential}
   In our analysis we show that our algorithm is correct under a sequential Gouda fair scheduler but from \cite[Theorem 8]{Gouda2001theory} our result also holds under a Gouda fair scheduler (that is not necessarily sequential). 
\end{remark}
  
We now present some notations and observations about the geometry of the system. 
Let $v$ and $v'$ be two neighbouring nodes in $\mathcal{P}$. We say that an edge that is oriented from $v$ to a neighbouring node $v'$ is \textit{outgoing} for $v$ and \textit{incoming} for $v'$. We write $\overrightarrow{vv'}$ to denote an edge directed from $v$ to $v'$ and $vv'$ to denote an undirected edge or an edge whose orientation is not known.
Particles with at least one neighbour that is not in $\mathcal{P}$ are
on the \textit{boundary}. Since $\mathcal{P}$ is simply connected,
there exists only one boundary in the system. Let $p$ be a particle on
the boundary. We say that $p$ is \emph{pending} if $p$ has a unique
neighbouring particle in $\cP$. We say that $p$ is an
\textit{articulation point} if the removal of $p$ disconnects
$\mathcal{P}$. If $p$ is neither pending, nor an articulation point,
then $p$ is incident to two distinct edges $pq$, $pr$ on the boundary
of $\cP$. In this case, since $\mathcal{P}$ is simply connected, there is a path of particles in the 1-neighbourhood of $p$ from $q$ to $r$. 
We say that $p$ is on a
$\theta \in \{60^\circ, 120^\circ, 180^\circ, 240^\circ\}$ angle to
denote the angle that is formed when moving from $q$ to $r$ around $p$
and no empty nodes are encountered. By slight abuse of notation, we
also call a particle on a $\theta$ angle a \textit{$\theta$ particle}.

It is easy to see that a particle on the boundary cannot be on a
$300^\circ$ angle, otherwise $q$ and $r$ are adjacent and $p$ is not
on the boundary, a contradiction. $\mathcal{P}$ is 2--connected if it does not contain any articulation point. Notice that in systems with at least three particles, a system with no articulation point does not contain
any pending particle. In a 2--connected particle system, the following observation implies that there should be a $60^\circ$ or a $120^\circ$
particle.

\begin{observation}\label{obs:formula}
  If $\mathcal{P}$ is 2--connected and contains at least three particles, particles on the boundary satisfy
  the formula $2n_{60} + n_{120} - n_{240} = 6$, where $n_{\theta}$ is
  the number of $\theta$ particles on the boundary.
\end{observation}

\begin{proof}
  If $\mathcal{P}$ is 2--connected, the boundary forms a simple
  polygon. We know that the sum of internal angles of a simple polygon
  is $(n-2)\pi$, where $n$ is the number of vertices of the
  polygon. So
  $(n_{60} + n_{120} + n_{180} + n_{240} -2)\pi = n_{60}\frac{\pi}{3}
  + n_{120}\frac{2\pi}{3} + n_{180}\pi + n_{240}\frac{4\pi}{3}$, that
  is, $2n_{60} + n_{120} - n_{240} = 6$.
\end{proof}

\begin{lemma}\label{lemma:boundaryparticles}
  In any simply connected particle system $\cP$ with at least two
  particles, the boundary of $\cP$ contains one of the following:
  \begin{enumerate}
  \item \label{item:pending} a pending particle, or
  \item \label{item:60} a $60^\circ$ particle, or
  \item \label{item:consecutive-120} two $120^\circ$ particles that are connected by a path of
    $180^\circ$ particles on the boundary.
  \end{enumerate}
\end{lemma}
\begin{proof}
  A block is a 2--connected component of $\cP$. As $\cP$ contains
  at least two particles, each block is either an edge or it contains
  at least three particles. The block tree of $\cP$ is a tree where
  each vertex is a block and there is an edge between two blocks if
  they share a vertex (i.e., an articulation point of
  $\cP$). A leaf, $\cP'$, of the block tree is a 2-connected component of $\cP$ and contains a unique articulation
  point $p'$ of $\cP$. In the remainder of the proof we show that there is a particle satisfying one of the cases of the lemma statement in every leaf. If $\cP'$ contains precisely two particles $p'$
  and $q'$, then $p'$ is the unique neighbour of $q'$ in $\cP$ and
  $q'$ is a pending particle, as in Case \ref{item:pending}.

  Suppose $\cP'$ contains at least three particles. Since
  $\cP'$ is 2-connected, every particle on the boundary of $\cP'$ is
  a $\theta \in \{60^\circ, 120^\circ, 180^\circ, 240^\circ\}$ particle. Any $\theta$ particle $p \neq p'$ of $\cP'$ is also a $\theta$
  particle of $\cP$. So a $60^{\circ}$ particle
  $p \neq p'$ in $\cP'$, is a $60^{\circ}$ particle in
  $\cP$, which is Case \ref{item:60}.

  Suppose now that in $\cP'$, any boundary particle $p$ different from
  $p'$ is a $\theta$ particle with
  $\theta \in \{120^\circ, 180^\circ, 240^\circ\}$. Let $n'_{120}$,
  $n'_{180}$, $n'_{240}$ be respectively the number of
  $120^\circ$, $180^\circ$, $240^\circ$ particles in $\cP'$ that are
  different from $p'$. Since $p'$ is an articulation point, $p'$
  cannot have more than three consecutive particle
  neighbours. Consequently, in $\cP'$, $p'$ is either a $60^\circ$ or
  a $120^\circ$ particle.  If $p'$ is a $60^\circ$ particle in $\cP'$,
  from Observation~\ref{obs:formula}, we have
  $2 + n'_{120} - n'_{240} = 6$.  If $p'$ is a $120^\circ$ particle in
  $\cP'$, from Observation~\ref{obs:formula}, we have
  $n'_{120} +1 - n'_{240} = 6$. In both cases, we then have
  $n'_{120} \geq n'_{240} + 4$.  Let $p_1, \ldots, p_{n'_{120}}$ be
  the $120^\circ$ particles of $\cP'$ in the order in which they
  appear when we move on the boundary of $\cP'$ starting from $p'$
  (i.e., $p'$ appears between $p_{n'_{120}}$ and $p_1$). Since
  $n'_{120} \geq n'_{240} + 4 > n'_{240} + 1$, there exists an index
  $1 \leq i \leq n'_{120}-1$ such that only $180^\circ$ particles
  appear on the boundary of $\cP'$ between $p_i$ and $p_{i+1}$. Since
  all these $180^\circ$ particles are also $180^\circ$ particles on
  the boundary of $\cP$, we are in Case \ref{item:consecutive-120}.
\end{proof}

We now explain how two adjacent particles in a triangle can detect each other's chirality. In the non self--stabilising setting, chirality detection is possible when movement is allowed as shown in \cite{emek2019deterministic} and in certain cases when particles agree on the directions along one axis of the grid \cite{chalopin2024deterministic}. Furthermore, particles in exactly one common triangle and no other neighbours participating in the procedure, can agree on chirality when they are able to exchange messages like in \cite{diluna2020shape}. To the best of our knowledge, when particles cannot move or exchange messages, like in our setting, a chirality detection algorithm does not exist even in the non self--stabilising context. 
For a particle $p$, we let $\{\msp_i \mid 0\leq i \leq 5\}$ be the set of ports incident to $p$. The label $\lambda(\Pi)$ of a path $\Pi = (p_1,p_2, \ldots, p_k)$ in
the graph induced by the particles is a sequence of pairs of labels
$(\msa_1,\msb_2),(\msa_2,\msb_3), \ldots, (\msa_{k-1},\msb_k)$ where
for each $i$, $\msa_i$ (resp. $\msb_i$) is the port connecting $p_i$
to $p_{i+1}$ (resp. $p_{i-1}$). Following the definitions of \cite{yamashita1988computing}, we define
$\view_k(p)$, to be the set of labels $\lambda(\Pi)$ of paths $\Pi$
starting at $p$ of length at most $k$. Note that for each $1 \leq j \leq k$,
if both $(\msa_1,\msb_2),(\msa_2,\msb_3), \ldots, (\msa_{j},\msb_{j+1})$ and
$(\msa_1,\msb_2),(\msa_2,\msb_3), \ldots, (\msa_{j},\msb_{j+1}')$ belong
to $\view_k(p)$,
then $\msb_{j+1} =\msb_{j+1}'$.
From \cite{boldi2002universal}, for any constant $k$, each particle
$p$ can construct $\view_k(p)$ in a self stabilising way with constant
memory.

As we have defined particles to be anonymous, particle names in port labels are used solely to improve readability (i.e., they do not correspond to particle identifiers). 
%An example of $\view_3(p)$ for particle $p$ in Figure \ref{fig:chirality-q2-up}, assuming that $qr$ is labelled $(\msq_2, \msr_0)$ and $qr'$ is labelled $(\msq_0,\msr'_2)$, is \{ [(1,1), (2,0), (1,0)], [(0,1), (0,2), (1,1)], [(1,1), (0,2)] \}. The difficulty in determining the chirality of a neighbouring particle therefore arises from $p$ needing to determine which particles the ports incident to $q$ correspond to.
An example of some of the paths belonging to $\view_3(p)$ for particle $p$ in Figure~\ref{fig:chirality-q2-up}, assuming that $qr$ is labelled $(\msq_2, \msr_0)$, $qr'$ is labelled $(\msq_0, \msr'_2)$ and for each $i$, $p_i \allowbreak = \allowbreak q_i \allowbreak = \allowbreak r_i \allowbreak = \allowbreak r'_i \allowbreak = i$, is: $[(1,1), \allowbreak (2,0), \allowbreak (1,0)], \allowbreak [(0,1), \allowbreak (0,2), \allowbreak (1,1)], \allowbreak [(1,1), (0,2)]$. The difficulty in determining the chirality of a neighbouring particle therefore arises from $p$ needing to determine which particles the ports incident to $q$ correspond to.

\begin{lemma}\label{lem:chirality-detection}
  For any triangle of particles $pqr$, $p$ can infer the chirality of
  $q$ from $\emph{\view}_3(p)$.
\end{lemma}
\begin{proof}
      In the following, for a particle $p$, we let
      $\{\msp_i \mid 0\leq i \leq 5\}$ be the set of ports incident to $p$
      and we assume that either $\msp_{i+1} = \msp_i + 1$ for each
      $0 \leq i \leq 5$, or $\msp_{i+1} = \msp_i - 1$ for each $0 \leq i \leq 5$
      (where additions are made modulo $6$). %We will use the following observation.

      Consider a triangle $pqr$. Let $\msp_1$ (resp. $\msq_1$) be the port
      through which $p$ (resp. $q$) is connected to $q$ (resp. $p$).
      Further, let $p$ (resp. $r$) be connected to $r$ (resp. $p$) through
      $\msp_0$ (resp. $\msr_1$). Observe that if $p$ learns the port
      through which $q$ is connected to $r$, it also learns the chirality
      of $q$. Note that by Observation~\ref{obs:consec-ports}, this port is
      either $\msq_0$ or $\msq_2$ and the port from $r$ to $q$ is either
      $\msr_0$ or $\msr_2$. Notice that if $r$ is the only common
      neighbour of $p$ and $q$, then only one of
      $\{(\msp_1,\msq_1),(\msq_0,\msx) \mid 0 \leq \msx \leq 5\} \cup
      \{(\msp_1,\msq_1),(\msq_2,\msx) \mid 0 \leq \msx \leq 5\}$ is in
      $\view_3(p)$ and $p$ can then infer the chirality of $q$. Suppose
      now that $p$ and $q$ have two common neighbours $r$ and $r'$.

      \begin{claim}\label{claim:chirality}
        Let $pqr$ be a triangle of particles. The edge $qr$ is labelled $(\msq_2,\msr_0)$ if and only if the following
        formula holds:
      
        $$ (\msp_1, \msq_1), (\msq_2, \msr_0), (\msr_1, \msp_0) \in \textup{\texttt{view}}_3(p) \text{ } \land $$
        $$ (\msp_0,\msr_1), (\msr_0, \msq_2), (\msq_1, \msp_1) \in \textup{\texttt{view}}_3(p) \text{ } \land$$
        $$ \Bigl[ (\msp_1, \msq_1), (\msq_0, \msr_2) \notin \textup{\texttt{view}}_3(p) \text{ } \lor $$
        $$ (\msp_2, \msr_5) \notin \textup{\texttt{view}}_3(p) \text{ } \Bigr] $$
      \end{claim}
      \begin{proof}
        Notice that the configuration is either the one on Figure \ref{fig:chirality-q2-up} or on Figure \ref{fig:chirality-q2-down}, since $p$ knows the labels of its incident edges.
      
        \begin{figure}[ht]
          \centering
            \begin{subfigure}{.49\textwidth}
              \centering
              \begin{tikzpicture}[scale=.4]
            
                  \triangularGrid{1}{2}{4}{2*1.73}
      
                  \filldraw[fill=white] (60:4) circle(10pt) node[] {\scriptsize $p$} node[label={[label distance=1mm]0:\scriptsize $\msp_1$}] {} node[label={[label distance=.5mm]80:\scriptsize $\msp_0$}] {} node[label={[label distance=.5mm]-80:\scriptsize $\msp_2$}] {} ++ (60:4) circle(10pt) node[] {\scriptsize $r$} node[label={[label distance=.5mm]-93:\scriptsize $\msr_1$}] {} ++ (-60:4) circle(10pt) node[] {\scriptsize $q$} node[label={[label distance=1mm]180:\scriptsize $\msq_1$}] {} node[label={[label distance=.5mm]93:\scriptsize $\msq_2$}] {}  ++ (-120:4) circle(10pt) node[] {\tiny $r'$};
              
              \end{tikzpicture}
              \caption{\centering}
                \label{fig:chirality-q2-up}
            \end{subfigure}
            \begin{subfigure}{.49\textwidth}
                \centering
                \begin{tikzpicture}[scale=.4]
              
                \triangularGrid{1}{2}{4}{2*1.73}
      
                \filldraw[fill=white] (60:4) circle(10pt) node[] {\scriptsize $p$} node[label={[label distance=1mm]0:\scriptsize \color{black} $\msp_1$}] {} node[label={[label distance=.5mm]80:\scriptsize $\msp_0$}] {} node[label={[label distance=.5mm]-80:\scriptsize $\msp_2$}] {} ++ (60:4) circle(10pt) node[] {\scriptsize $r$} node[label={[label distance=.5mm]-93:\scriptsize $\msr_1$}] {} ++ (-60:4) circle(10pt) node[] {\scriptsize $q$} node[label={[label distance=1mm]180:\scriptsize $\msq_1$}] {} node[label={[label distance=.5mm]-93:\scriptsize $\msq_2$}] {} ++ (-120:4) circle(10pt) node[] {\tiny $r'$};
                
                \end{tikzpicture}
                \caption{\centering}
                \label{fig:chirality-q2-down}
            \end{subfigure}
            \caption{The 1-view of $p$ and the two possible orientations of $q$}
        \end{figure}
        
        First let us suppose that the edge $qr$ is labelled
        $(\msq_2,\msr_0)$. Then, the first two expressions of the formula
        are satisfied. Let us suppose
        $(\msp_1, \msq_1), (\msq_0, \msr_2) \in \view_3(p)$. Then from
        Observation \ref{obs:consec-ports}, $qr'$ is labelled $(\msq_0,\msr_2)$
        and $pr'$ is either labelled $(\msp_2,\msr_1)$ or
        $(\msp_2,\msr_3)$. In either case,
        $(\msp_2,\msr_5) \notin \view_3(p)$ and the formula is satisfied.
      
        Let us now suppose that the formula is satisfied and assume that
        $qr$ is not labelled $(\msq_2,\msr_0)$. Then by
        Observation~\ref{obs:consec-ports}, $qr$ is labelled either
        $(\msq_2,\msr_2)$, or $(\msq_0,\msr_0)$, or
        $(\msq_0,\msr_2)$. Note that the first two cases are impossible
        since $(\msp_1,\msq_1),(\msq_2,\msr_0)$ and
        $(\msp_0,\msr_1),(\msr_0,\msq_2)$ belong to $\view_3(p)$.
        Consequently, $qr$ is labelled $(\msq_0,\msr_2)$ and since
        $(\msp_1,\msq_1)(\msq_2,\msr_0) \in \view_3(p)$, by Observation
        \ref{obs:consec-ports}, $qr'$ is labelled $(\msq_2,\msr_0)$, and
        the label of $pr'$ is either $(\msp_2,\msr_5)$ or
        $(\msp_2,\msr_1)$. Note that we are necessarily in the second case
        since we assumed that the formula holds and since
        $(\msp_1,\msq_1)(\msq_0,\msr_2) \in \view_3(p)$. This implies that
        $(\msp_1,\msq_1)(\msq_2,\msr_0)(\msr_1,\msp_2) \in \view_3(p)$,
        and thus
        $(\msp_1,\msq_1)(\msq_2,\msr_0)(\msr_1,\msp_0) \notin \view_3(p)$,
        contradicting the fact that the formula holds. This concludes the proof of Claim \ref{claim:chirality}.        
      \end{proof}  
    
    From Observation \ref{obs:consec-ports}, $qr$ must be labelled
    $(\msq_0,\msr_0)$, $(\msq_0,\msr_2)$, $(\msq_2,\msr_0)$ or $(\msq_2,\msr_2)$. Applying Claim
    \ref{claim:chirality} to each possibility, $p$ can detect the label
    of $qr$ and thus infer the chirality of $q$, concluding the proof of Lemma \ref{lem:chirality-detection}.
\end{proof}

Notice that the first condition of Claim \ref{claim:chirality} does not imply the second. Figure \ref{fig:chirality-cond1-3} provides an example of a port labeling that satisfies the first and third conditions, while the edge $qr$ is not labelled $(\msq_2, \msr_0)$.

\begin{figure}[ht]
  \centering
    \begin{tikzpicture}[scale=.4]
  
        \triangularGrid{1}{2}{4}{2*1.73}

        \filldraw[fill=white] (60:4) circle(10pt) node[] {\scriptsize $p$} node[label={[label distance=1mm]0:\scriptsize $\msp_1$}] {} node[label={[label distance=.5mm]80:\scriptsize $\msp_0$}] {} node[label={[label distance=.5mm]-80:\scriptsize $\msp_2$}] {} ++ (60:4) circle(10pt) node[] {\scriptsize $r$} node[label={[label distance=.5mm]-93:\scriptsize $\msr_1$}] {} node[label={[label distance=.5mm]-60:\scriptsize $\msr_0$}] {} ++ (-60:4) circle(10pt) node[] {\scriptsize $q$} node[label={[label distance=1mm]180:\scriptsize $\msq_1$}] {} node[label={[label distance=.5mm]93:\scriptsize $\msq_0$}] {} node[label={[label distance=.5mm]-95:\scriptsize $\msq_2$}] {}  ++ (-120:4) circle(10pt) node[] {\tiny $r'$} node[label={[label distance=.5mm]95:\scriptsize $\msr_5$}] {} node[label={[label distance=.5mm]80:\scriptsize $\msr_0$}] {} node[label={[label distance=.5mm]0:\scriptsize $\msr_1$}] {} ++ (0:4) circle(10pt) node[] {\tiny $p'$} node[label={[label distance=.5mm]180:\scriptsize $\msp_0$}] {}; 
    
    \end{tikzpicture}
    \caption{\centering A particle configuration satisfying the first and last condition of Claim \ref{claim:chirality}, but not the second.}
    \label{fig:chirality-cond1-3}
\end{figure}

\section{A Proof Labelling Scheme  for Leader Election}\label{sec:certificate}
Our aim is to orient the edges between particles so that
a unique sink particle (i.e., particle with no outgoing edges) that
we define to be the leader exists. The certificate given to each particle
consists of a direction for each edge incident to the particle. The
orientation of the edges is chosen so that particles that are reached
by an outgoing edge of some particle $p$, induce a connected graph of
size at most 3.
In general we cannot avoid the existence of directed cycles in the
orientation, but we will show that the existence of a unique sink is
always guaranteed. Each particle $p$ can check that the following
rules are locally satisfied or detect an error.

\begin{enumerate}[label=R\arabic*]
    \item \label{rule:all-edges-directed} Each edge is oriented and both particles agree on the direction of the edge. 
    \item \label{rule:incoming-edge} Particle $p$ has at most three outgoing edges. We consider edges between $p$ and empty nodes to be incoming for $p$.
    \item \label{rule:consecutive} When looking at the ports of $p$ cyclically, all outgoing edges of $p$ are consecutive. 
    \item \label{rule:no-cyclic-triangles} For every 3-particle triangle $p$ belongs to, the triangle is not a directed cycle. 
\end{enumerate}

In Figure \ref{fig:rules-broken} we give and example of each rule being violated.
\begin{figure}[htb]
  \centering
  \includegraphics[scale=.59,page=1]{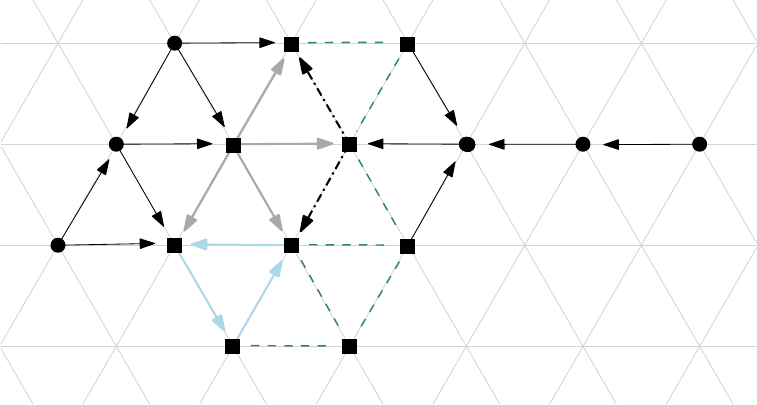}
  \caption{A configuration containing broken rules. Each square particle can detect at least one rule being violated. Green (dashed) edges break \ref{rule:all-edges-directed}, grey edges break \ref{rule:incoming-edge}, black (dash dotted) edges break \ref{rule:consecutive} and blue edges break \ref{rule:no-cyclic-triangles}.}
  \label{fig:rules-broken}
\end{figure}

We call a configuration where every particle satisfies
\ref{rule:all-edges-directed}--\ref{rule:no-cyclic-triangles} a
\textit{valid configuration}.  Note that
\ref{rule:no-cyclic-triangles} does not guarantee an acyclic
orientation (i.e., that larger directed cycles do not exist in the
configuration). We do not forbid global directed cycles, but we will prove that
even if directed cycles of size larger than three are formed by the incoming
and outgoing edges, the remaining rules guarantee that there exists a
unique sink in the system that we define to be the leader. An example of a directed cycle of size larger than three that does not violate any rule is given in Figure \ref{fig:global-cycle}.

\begin{figure}[t]
  \centering
  \includegraphics[scale=.8,page=1]{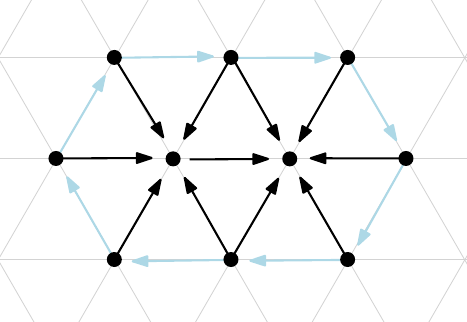}
  \caption{An example of a configuration containing a directed cycle (marked by blue edges) of more than three particles that does not violate any rule.}
  \label{fig:global-cycle}
\end{figure}

\begin{theorem}\label{th:certificate}
  If all rules
  \ref{rule:all-edges-directed}--\ref{rule:no-cyclic-triangles} are
  satisfied, then there exists a unique sink in the system.
\end{theorem}
% To prove Theorem \ref{th:certificate} we need to show two properties:
% \begin{enumerate}
%   \item It is possible to orient all edges so that all rules are satisfied.
%   \item An orientation satisfying \ref{rule:all-edges-directed}--\ref{rule:no-cyclic-triangles} always guarantees a unique leader in the system.
% \end{enumerate}
% We show the former property here and we prove the second property in Section \ref{sec:proof}. 
We prove Theorem \ref{th:certificate} in Section \ref{sec:proof}. 
A valid configuration that does not contain any directed cycle is a
\textit{valid acyclic orientation}. Observe that any particle system
admits a valid acyclic orientation, as it can be constructed from any
erosion based Leader Election algorithm for programmable matter
(e.g., \cite{diluna2020shape,dufoulon2021efficient}). Indeed, consider
an execution of an erosion algorithm on a system, and orient any edge
$pq$ from $p$ to $q$ if $p$ is eroded before $q$ in the
execution. This orientation is acyclic and it thus obviously satisfies
\ref{rule:all-edges-directed} and \ref{rule:no-cyclic-triangles}.
Since an erosion based algorithm erodes only a particle that does not
disconnect its neighbourhood and that is strictly convex (i.e., that
has at most three consecutive non-eroded neighbours), the orientation also
satisfies \ref{rule:incoming-edge} and \ref{rule:consecutive}. 

\section{A Self-Stabilising Algorithm for Leader Election} \label{sec:algorithm}

In Section \ref{sec:certificate}, we claimed that \emph{if} all rules are locally satisfied a unique sink exists in the system. Here, we show \emph{how} a valid configuration is reached from a configuration containing errors. Our algorithm is simple: when a particle, $p$, is incident to an undirected edge $e$, $p$ orients $e$ as outgoing. If the orientation of the edges incident to $p$ violates a rule, $p$ undirects all its outgoing edges. Each activated particle always executes both lines of Algorithm \ref{algo:ss-le}.  

\begin{algorithm}[ht]
    \caption{Self Stabilising LE}
    \label{algo:ss-le}

    \If{$\neg$\emph{\ref{rule:all-edges-directed}}}{Mark all \texttt{undirected} edges as \texttt{outgoing}}
    \If{$\neg$\emph{\ref{rule:incoming-edge}} $\lor$ $\neg$\emph{\ref{rule:consecutive}} $\lor$ $\neg$\emph{\ref{rule:no-cyclic-triangles}}}{Mark all \texttt{outgoing} edges as \texttt{undirected}}

\end{algorithm}

The directed edges incident to particles are encoded by each particle $p$
having a variable \emph{link}$_p[v'] \in \{in, out\}$ for each
neighbouring node $v'$. 
For particle $p \in \mathcal{P}$ and
$v' \in V_{G_\Delta \backslash \mathcal{P}}$,
\emph{link}$_p[v'] = in$. Any particle $p$ can locally detect whether
it is incident to an empty node, so we assume that edges between
occupied and unoccupied nodes are always marked correctly. For two adjacent particles
$p,p' \in \mathcal{P}$, if
\emph{link}$_p[p'] = in$ and \emph{link}$_{p'}[p] =
out$, $pp'$ is directed from $p'$ to $p$. We encode an \textit{undirected} edge between two particles $p$
and $p'$ as \emph{link}$_p[p'] =$ \emph{link}$_{p'}[p] = in$. We
address \emph{link}$_p[p'] =$ \emph{link}$_{p'}[p] = out$, as a
special case. The endpoint that is activated first (say $p$) marks
\emph{link}$_p[p']$ as $in$. Notice that this is only possible during
the first activation of $p$. Take for example each edge $pp'$. Executing Algorithm \ref{algo:ss-le}, $p$ can either orient edges as outgoing (by setting \emph{link}$_{p}[p']$ from $in$ to $out$) or mark edges as undirected (by setting \emph{link}$_{p}[p']$ from $out$ to $in$). Notice that in the former case, since the edge is marked as \emph{outgoing}, \emph{link}$_{p'}[p]$ must be $in$. Hence, after the activation of a particle $p$ for each edge $pp'$, it is not possible to have \emph{link}$_{p}[p']$ = \emph{link}$_{p'}[p] = out$.

In the remainder of this paper we only
use the orientation of the edges without referencing their encoding. That is, we say that an edge between two
particles $p$ and $q$ is: directed from $p$ to $q$ (i.e.,
$\overrightarrow{pq}$), directed from $q$ to $p$ (i.e.,
$\overrightarrow{qp}$) or undirected (i.e., $pq$ or $qp$). From Lemma \ref{lem:chirality-detection} particles in a common triangle detect each other's chirality. Since particles know both labels assigned to an edge, particles can compute the orientation of edges in triangles they belong to and check \ref{rule:no-cyclic-triangles}. From now on we only refer to particles detecting cyclic triangles. We prove that when executing our algorithm, any particle system reaches a valid configuration that contains a unique sink.

\begin{theorem}\label{th:algo-correct}
  % Starting from an arbitrary simply connected configuration any Gouda fair execution of Algorithm \ref{algo:ss-le} eventually
  % reaches a configuration satisfying
  % \ref{rule:all-edges-directed}--\ref{rule:no-cyclic-triangles} in which no rules can be applied and there exists a unique sink.
  Starting from an arbitrary simply connected configuration any Gouda fair execution of Algorithm \ref{algo:ss-le} eventually
  reaches a configuration in which no rules can be applied and all of  \ref{rule:all-edges-directed}--\ref{rule:no-cyclic-triangles} are satisfied. Any configuration satisfying all
  \ref{rule:all-edges-directed}--\ref{rule:no-cyclic-triangles} contains a unique sink.
\end{theorem}

\section{Proof of Theorem~\ref{th:certificate} and Theorem~\ref{th:algo-correct}}\label{sec:proof}
Here, we prove Theorem~\ref{th:algo-correct}. Notice
the second statement of Theorem~\ref{th:algo-correct} is precisely
Theorem~\ref{th:certificate}. From Remark \ref{rem:sequential}, we can consider sequential Gouda fair executions in the remainder of the paper. A configuration of a particle system
executing Algorithm~\ref{algo:ss-le} is described by the direction of each edge $pq$ (i.e., $\overrightarrow{pq}$,
$\overrightarrow{qp}$ or undirected). We make a few observations on how to change the orientation of some edges of a valid configuration and maintain a valid configuration.

\begin{observation}\label{obs:edge-change-dir-consecutive}
  Let $p$ be a particle such that \ref{rule:consecutive} is satisfied
  at $p$. Let $e$ be an incoming edge to $p$ and $e'$ be an outgoing
  edge of $p$, s.t. when moving cyclically around $p$, $e$ and
  $e'$ are consecutive. If $e$ becomes outgoing (resp. $e'$ becomes
  incoming), \ref{rule:consecutive} is not violated at $p$.
\end{observation}

\begin{observation}\label{obs:r1-r2-r4-hold}
  Let $C$ be a configuration and let $p$ be a particle so that
  \ref{rule:all-edges-directed} (resp., \ref{rule:incoming-edge},
  \ref{rule:no-cyclic-triangles}) is satisfied at some particle $q \neq p$
  in $C$. Then, \ref{rule:all-edges-directed} (resp.,
  \ref{rule:incoming-edge}, \ref{rule:no-cyclic-triangles}) is satisfied at $q$ in
  $C \setminus \{p\}$.
\end{observation}

Let $\mathcal{S} = C_0, C_1,\ldots$ be an execution of Algorithm
\ref{algo:ss-le} starting from a configuration $C_0$. Notice a
particle $p$ is activable in a configuration $C$ if when it executes
Algorithm~\ref{algo:ss-le}, its undirected edges
become outgoing or its outgoing edges become undirected. If
there exists a configuration $C_f$ where no particle is activable, then
$C_f = C_j$ for all $j > f$, and we say that the execution stabilises
to a \textit{final} configuration. If all rules are satisfied in this
final configuration, then this configuration is valid and we say it is
a \emph{final directed} configuration. If a configuration $C_i$ is not
final, we can assume that there exists an activable particle $p_i$
such that we obtain $C_{i+1}$ by activating $p_i$ in $C_i$.  As mentioned before, we consider finite systems of particles where the number of possible configurations is finite. Hence there exists an index $i_0$ in $\mathcal{S}$
such that any configuration $C_{i}$ with $i \geq i_0$ appears infinitely often
in $\mathcal{S}$. We write
$\mathcal{S}_{i_0} = C_{i_0}, C_{i_0+1},\ldots$ to denote the part of
the execution starting at $C_{i_0}$ and in the following we consider only 
$\cS_{i_0}$ and configurations $C_i$ with $i \geq i_0$. 
We call the edges that are
never undirected in $\mathcal{S}_{i_0}$, \textit{stable
edges}. Observe that by the definition of $i_0$, each edge is either
stable or undirected infinitely often. Notice that any edge $pq$ directed
from $p$ to $q$ in $C_{i}$ with $i \geq i_0$, is directed
from $p$ to $q$ infinitely often, regardless of whether it is stable. 

We first establish a few properties satisfied in $\cS_{i_0}$.

\begin{lemma}\label{lemma:hated-lemma}
  In any configuration $C_i$ with $i \geq i_0$, Rules
  \ref{rule:incoming-edge}, \ref{rule:consecutive}
  and \ref{rule:no-cyclic-triangles} are always satisfied between particle activations.
\end{lemma}

\begin{proof}
  Let $n_i$ be the number of particles in $C_i$ that do not satisfy
  \ref{rule:incoming-edge}, \ref{rule:consecutive} or
  \ref{rule:no-cyclic-triangles}.  If a particle $p$ is activated in
  $C_i$, then \ref{rule:incoming-edge}, \ref{rule:consecutive} and
  \ref{rule:no-cyclic-triangles} are satisfied at $p$ in
  $C_{i+1}$. Moreover, if \ref{rule:incoming-edge},
  \ref{rule:consecutive}, and \ref{rule:no-cyclic-triangles} are
  satisfied at some particle $p$ in $C_i$ that is not activated at
  step $i$, then they are still satisfied at step
  $C_{i+1}$. Consequently, $n_{i+1} \leq n_i$.  Since for
  $i \geq i_0$, $C_i$ appears infinitely often, we get that for every
  $i \geq i_0$, we have $n_i = n_{i_0}$. If $n_{i_0} > 0$, there exists a
  particle $p$ that violates one of the rules
  \ref{rule:incoming-edge}, \ref{rule:consecutive} or
  \ref{rule:no-cyclic-triangles} in $C_i$ for all $i \geq i_0$. Thus, $p$ is eventually activated
  at some step $i \geq i_0$ and in $C_{i+1}$, $p$ satisfies the rules, a
  contradiction. Consequently, for any $i \geq i_0$,
  \ref{rule:incoming-edge}, \ref{rule:consecutive} and
  \ref{rule:no-cyclic-triangles} are satisfied at every particle in
  $C_i$.
\end{proof}

\begin{lemma}\label{lemma:stable-particles}
  If a particle $p$ is incident to a stable outgoing edge, $p$ is
  never activable in $\mathcal{S}_{i_0}$ and all edges incident to $p$
  are stable edges.
\end{lemma}

\begin{proof}
  In a configuration $C_i$, if a particle $p$ is incident
  to an outgoing edge and an undirected edge, then $p$ is activable in
  $C_i$. After its activation, either all the undirected edges
  incident to $p$ have become outgoing edges, or all outgoing edges of
  $p$ have become undirected.
   
  Let $\overrightarrow{pq}$ be a stable edge, hence, $p$ never marks
  $\overrightarrow{pq}$ as undirected. Let us suppose that in addition
  to $\overrightarrow{pq}$, $p$ is also incident to an unstable edge
  $pr$. Then infinitely often, $pr$ is undirected and thus there
  exists a step where $p$ is activated and $pr$ is undirected. At this
  step, $p$ marks $\overrightarrow{pr}$ as outgoing. Then if
  $\overrightarrow{pr}$ becomes undirected at a later step,
  $\overrightarrow{pq}$ must also become undirected, which is a
  contradiction. Hence, all edges incident to $p$ are stable.
\end{proof}

\begin{lemma}\label{lemma:unloved-lemma}
  If a particle $p$ is incident to an unstable edge in $\mathcal{S}_{i_0}$, the unstable
  edges incident to $p$ are at least two and do not appear consecutively around
  $p$, or there are at least four unstable edges incident to $p$. 
\end{lemma}

\begin{proof}
  Suppose the lemma does not hold, implying that there exists a particle $p$
  incident to $1 \leq k \leq 3$ unstable edges that appear
  consecutively around $p$. Then, there exists an unstable edge $pq$
  incident to $p$ such that for every unstable edge $pr$ incident to
  $p$, either $r = q$ or $r$ is adjacent to $q$.  Note that by
  Lemma~\ref{lemma:stable-particles}, all stable edges incident to $p$
  are incoming to $p$.
  
  \begin{claim}\label{claim:unloved-claim}
    Consider an unstable edge $pr$ with $r \neq q$ and let $s$ be the
    common neighbour of $p$ and $r$ that is distinct from $q$. Then
    $\overrightarrow{sp}$ and $\overrightarrow{sr}$ are stable.
  \end{claim}

  \begin{proof}
    Since $s$ is not adjacent to $q$, ${sp}$ is stable, and by
    Lemma~\ref{lemma:stable-particles}, $sp$ is oriented from $s$ to
    $p$. By Lemma~\ref{lemma:stable-particles} applied at $s$ and $r$,
    $sr$ is also stable and it is directed from $s$ to $r$.
  \end{proof}

  Suppose first there is a configuration $C_i$ with $i\geq i_0$
  such that $pq$ is directed from $p$ to $q$ in $C_i$ and undirected
  in $C_{i+1}$. This implies that $p$ is activated at step $i$. By Lemma~\ref{lemma:hated-lemma} there exists at least one undirected
  edge $pr$ in $C_i$, and from the second condition of Algorithm \ref{algo:ss-le}, when orienting all undirected edges incident
  to $p$ in $C_i$ as outgoing edges, one of
  \ref{rule:incoming-edge}, \ref{rule:consecutive},
  \ref{rule:no-cyclic-triangles} is violated. Since $q$ is neighbouring to all other endpoints of unstable edges around $p$ (if more unstable edges than $pq$ exist around p), this cannot be \ref{rule:incoming-edge} or
  \ref{rule:consecutive}.  If \ref{rule:no-cyclic-triangles} is
  violated, it implies that in $C_i$, there exists an undirected edge
  $pr$ and directed edges $\overrightarrow{rs}$ and
  $\overrightarrow{sp}$. By Claim~\ref{claim:unloved-claim}, $s = q$
  but this is impossible since $\overrightarrow{pq}$ is in $C_i$.

  Then, at each step $i \geq i_0$, either $qp$ is undirected or it is
  directed from $q$ to $p$. If there is no step $i \geq i_0$ where $q$
  is activated, then $q$ never has any outgoing edge, $qp$ is always
  undirected and we let $i_1 = i_0$. Otherwise, consider a step
  $i_1-1$ where $q$ is activated such that in $C_{i_1}$, $qp$ is
  undirected. Then at step $i_1$, all edges incident to $q$ are either
  incoming or undirected. We claim that if we activate $p$ at step
  $i_1$, it orients $pq$ from $p$ to $q$. Indeed by the definition of
  $q$, rules \ref{rule:incoming-edge} and \ref{rule:consecutive} are
  satisfied when $p$ orients all its outgoing edges. By
  Claim~\ref{claim:unloved-claim}, any triangle violating
  \ref{rule:no-cyclic-triangles} should contain $q$, but this is
  impossible since $q$ has no outgoing edges in
  $C_{i_1}$. So, by the fairness condition, there exists a
  configuration $C_i$ containing $\overrightarrow{pq}$, a
  contradiction.
\end{proof}

We now prove Theorem~\ref{th:algo-correct} using the structure of the
boundary of $\cP$ given by Lemma~\ref{lemma:boundaryparticles}. Informally, the proof has the following structure. We assume that it is possible that the system does not stabilise and we will arrive at a contradiction. Out of the particle systems that do not stabilise to a configuration that satisfies all rules and has a unique sink, we take a particle system with the minimum number of particles. On the boundary of that system there exists a particle $p$ satisfying one of the cases of Lemma \ref{lemma:boundaryparticles}. For each orientation of the edges incident to $p$ we show that the edges incident to $p$ are stable. Then we take a smaller system that contains exactly one less particle, $p$. We show that the execution in both systems for particles that are not $p$ is the same and as a result, if the system that contains $p$ does not satisfy all rules and does not have a unique sink, then the same is true for the system that does not contain $p$. Since we had assumed that the system containing $p$ is the minimum size system that does not stabilise to a valid configuration, a smaller system not stabilising is a contradiction. 

 \begin{proof}[Proof of Theorem~\ref{th:algo-correct}]

  Let us suppose that there exists a fair execution
  $\cS=C_0,C_1,\ldots$ on a particle configuration $C=C_0$ that does
  not stabilise to a final directed configuration containing a
  unique sink. Consider such an execution $\cS$ with a support $\cP$
  of minimum size.  As defined above, consider a fair execution
  $\mathcal{S}_{i_0}= C_{i_0},C_{i_0+1},\ldots$ containing only
  configurations appearing infinitely often. By Lemma~\ref{lemma:boundaryparticles}, we can assume that the
  boundary of $\cP$ contains either a pending particle, or a
  $60^\circ$ particle, or two $120^\circ$ particles that are connected
  by a path of $180^\circ$ particles on the boundary. In the
  following, we show that each of these cases cannot occur. We first
  consider the case where $\cP$ contains a pending particle.

  \begin{lemma}\label{lemma:pending}
    If $\cP$ contains a pending particle $p$ (i.e., a particle with
    only one neighbouring particle $w$), all edges are stable in
    $\cS_{i_0}$ and there is a unique sink in the final configuration.
  \end{lemma}

  \begin{proof}
    By Lemma~\ref{lemma:unloved-lemma}, the edge $pw$ is stable.
    Suppose first that $pw$ is directed from $p$ to $w$ in
    $\cS_{i_0}$.  For each $i \geq i_0$, let
    $C'_i = C_i \setminus \{p\}$ and consider the sequence of
    configurations
    $\mathcal{S}'_{i_0} = C'_{i_0}, C'_{i_0+1}, \ldots, C'_{i},
    \ldots$. Observe that for each $i \geq i_0$, either $C_{i+1} = C_i$
    or there exists $p_i$ such that $C_{i+1}$ is obtained from $C_i$
    by activating $p_i$ and thus modifying the orientations of edges
    incident to $p_i$. Since $\overrightarrow{pw}$ is stable, by
    Lemma~\ref{lemma:stable-particles}, for any $i \geq i_0$,
    $p_i \neq p$. Moreover, for each $i \geq i_0$, the edges of $C_i'$
    have the same orientation as in $C_{i}$. So, $p' \neq p$ is
    activable in $C_i'$ if and only if it is activable in
    $C_i$. Furthermore, the configuration obtained by activating $p_i$
    in $C_i'$ is precisely $C_{i+1}'$ since the edges of $C_i'$ have
    the same orientation as in $C_{i}$. Hence,
    $\mathcal{S}'_{i_0}$ is a fair execution of Algorithm
    \ref{algo:ss-le} on $\mathcal{P}\setminus \{p\}$. By the
    minimality of the size of $\mathcal{P}$, there exists a step
    $i_1 \geq i_0$ such that $C_{i_1}'$ is a final configuration
    where all edges are oriented and that contains a unique sink
    $p''$. Since the edges incident to $p$ are stable, $C_{i_1}$ is a
    final configuration where all edges are oriented. By our
    definition of $i_0$, this implies that $i_1 = i_0$. Since $p$ has
    an outgoing edge, $\overrightarrow{pw}$, in $C_{i_1} = C_{i_0}$,
    $p$ is not a sink of $C_{i_0}$ and $p''$ is the
    unique sink in $C_{i_0}$.

    Suppose now that $\overrightarrow{wp}$ is stable. Notice that in
    this case, $p$ is a sink in $C_{i}$ for each $i \geq
    i_0$. Moreover, since $\overrightarrow{wp}$ is stable, by
    Lemma~\ref{lemma:stable-particles}, $w$ is never activated. Since
    the two common neighbours of $p$ and $w$ are empty, by
    \ref{rule:consecutive}, $p$ is the only outgoing neighbour of $w$
    in $C_i$ for any $i \geq i_0$. Consequently, $w$ is a sink in
    $C_i\setminus \{p\}$ for any $i \geq i_0$.  For each $i \geq i_0$,
    let $C'_i = C_i \setminus \{p\}$ and consider the sequence of
    configurations
    $\mathcal{S}'_{i_0} = C'_{i_0}, C'_{i_0+1}, \ldots, C'_{i},
    \ldots$. Observe that for each $i \geq i_0$, either $C_{i+1} = C_i$
    or there exists $p_i$ such that $C_{i+1}$ is obtained from $C_i$
    by activating $p_i$ and thus modifying the orientations of edges
    incident to $p_i$. Since $p$ has only incoming edges in $C_i$,
    $p_i \neq p$. Moreover, since $\overrightarrow{wp}$ is stable,
    $p_i \neq w$.  Furthermore, for each $i \geq i_0$, the edges of
    $C_i'$ have the same orientation as in $C_{i}$. Consequently,
    $p' \neq p$ is activable in $C_i'$ if and only if it is activable in
    $C_i$. Furthermore, the configuration obtained by activating $p_i$
    in $C_i'$ is precisely $C_{i+1}'$ since the edges of $C_i'$ have
    the same orientation as in $C_{i}$. Consequently,
    $\mathcal{S}'_{i_0}$ is a fair execution of Algorithm
    \ref{algo:ss-le} on $\mathcal{P}\setminus \{p\}$. By the
    minimality of the size of $\mathcal{P}$, there exists a step
    $i_1 \geq i_0$ such that $C_{i_1}'$ is a final configuration where
    all edges are oriented and that contains a unique sink $p''=w$ in
    $C_{i_1}'$. Since $\overrightarrow{wp}$ is stable, $C_{i_1}$ is a
    final configuration where all edges are oriented.  By our
    definition of $i_0$, this implies that $i_1 = i_0$. Since any
    $p' \in C_{i_0}\setminus \{p,w\}$ is not a sink in $C'_{i_0}$, and
    since $\overrightarrow{wp}$ is in $C_{i_0}$, $p$ is the unique
    sink in the valid configuration $C_{i_0}$.
  \end{proof}

  We now consider the case where the boundary of $\cP$ contains a
  $60^\circ$ particle $p$, and we let $q$ and $r$ be the two
  neighbours of $p$ on the boundary of $\cP$.

  \begin{lemma}\label{lemma:60-particles-stabilise}
    If $\cP$ contains a $60^\circ$ particle $p$, all edges are stable
    and there is a unique sink in the final configuration.
  \end{lemma}

  \begin{proof}
    By Lemma~\ref{lemma:unloved-lemma}, $pq$ and $pr$ are
    stable. Consequently, $pq$ and $pr$ are always directed in the same way all
    along $\cS_{i_0}$ and we can talk about the orientation of $pq$
    and $pr$ in $\cS_{i_0}$. For each $i \geq i_0$, let $C'_i = C_i \setminus \{p\}$ and consider the sequence of configurations
    $\mathcal{S}_{i_0}' = C'_{i_0}, C'_{i_0+1}, \ldots, C'_{i},
    \ldots$. For each $i \geq i_0$, either $C_{i+1} = C_i$ or there
    exists $p_i$ such that $C_{i+1}$ is obtained from $C_i$ by
    activating $p_i$ and thus modifying the orientations of edges
    incident to $p_i$. Since all edges incident to $p$ are stable in
    $\mathcal{S}_{i_0}$, we can assume $p_i \neq p$, for any $i \geq i_0$.

    We distinguish three cases, depending on the orientation of $pq$
    and $pr$ in $\cS_{i_0}$.

    \begin{case}\label{case:60-stabilize-out-out}
      The edges incident to $p$ are $\overrightarrow{pq}$ and
      $\overrightarrow{pr}$.
    \end{case}

    \begin{proof}
      For each $i \geq i_0$, the edges of $C_i'$ have the same
      orientation as in $C_{i}$. Hence, a particle $p' \neq p$
      is activable in $C_i'$ if and only if it is activable in
      $C_i$. The configuration obtained by activating
      $p_i$ in $C_i'$ is precisely $C_{i+1}'$ since the edges of
      $C_i'$ have the same orientation as in $C_{i}$. So,
      $\mathcal{S}'_{i_0}$ is a fair execution of Algorithm
      \ref{algo:ss-le} on $\cP\setminus \{p\}$. By the minimality of
      the size of $\cP$, there exists a step $i_1 \geq i_0$ such that
      $C_{i_1}'$ is a final configuration with a unique sink $p''$ where all edges are oriented. Since the edges
      incident to $p$ are stable, $C_{i_1}$ is a final configuration
      where all edges are oriented. By our definition of $i_0$, this
      implies that $i_1 = i_0$. Since $p$ has only outgoing edges in
      $C_{i_1} = C_{i_0}$, $p$ is not a sink of $C_{i_0}$ and $p''$ is
      the unique sink in $C_{i_0}$.
    \end{proof}

    \begin{case} \label{case:60-stabilize-in-out} The edges incident
      to $p$ are $\overrightarrow{qp}$ and $\overrightarrow{pr}$.
    \end{case}

    \begin{proof}
      Since $\overrightarrow{qp}$ is stable, $qr$ is also stable by
      Lemma~\ref{lemma:stable-particles}. By
      \ref{rule:no-cyclic-triangles}, $qr$ is directed from $q$ to
      $r$ in $\cS_{i_0}$. Notice that since $q$ and $p$ are incident
      to outgoing stable edges, from Lemma
      \ref{lemma:stable-particles}, $p$ and $q$ are incident only to
      stable edges and are never activable. The edges of
      $C'_{i \geq i_0} \setminus \{p\}$ have the same orientation as
      in $C_{i \geq i_0}$. Consequently, for any $p' \notin \{p,q\}$,
      $p'$ is activable in $C'_i$ if and only if it is activable in
      $C_i$.  Let us consider $q$. In $C_i$, $\overrightarrow{qp}$ is
      stable and directed and in $C'_i$, $q$ has an incoming edge from
      the respective empty node. Furthermore, $q$ is
      incident to an incoming edge from the empty common neighbour of
      $p$ and $q$. Hence, \ref{rule:consecutive} is satisfied for $q$
      in $C'_i$ from Observation
      \ref{obs:edge-change-dir-consecutive} and the remaining rules are
      satisfied for $q$ in $C'_i$ from Observation
      \ref{obs:r1-r2-r4-hold}. Therefore, $q$ is never activable in
      $C'_{i\geq i_0}$ and thus a particle $p'\neq p$ is activable in
      $C'_{i\geq i_0}$ if and only if it is activable in
      $C_{i\geq i_0}$. Consequently, $\mathcal{S}'_{i_0}$ is a fair
      execution of Algorithm \ref{algo:ss-le} on $\cP\setminus
      \{p\}$. By the minimality of the size of $\cP$, there exists a
      step $i_1 \geq i_0$ such that $C_{i_1}'$ is a stable
      configuration where all edges are oriented and that contains a
      unique sink $p''$. Note that $p'' \neq q$, since $\overrightarrow{qr} \in C_{i_1}'$. Since the edges incident to $p$ are stable
      and since $q$ is not activable, $C_{i_1}$ is a stable
      configuration where all edges are oriented. By our definition of
      $i_0$, this implies that $i_1 = i_0$. Since $p$ has an outgoing
      edge in $C_{i_1} = C_{i_0}$, $p$ is not a sink of $C_{i_0}$ and
      $p''$ is therefore the unique sink in $C_{i_0}$.
    \end{proof}
        
    \begin{case} \label{case:60-stabilize-in-in} The edges incident to
      $p$ are $\overrightarrow{qp}$ and $\overrightarrow{rp}$.
    \end{case}

    \begin{proof}
      Since $\overrightarrow{qp}$ and $\overrightarrow{rp}$ are
      stable, from Lemma~ \ref{lemma:stable-particles}, $q$ and $r$
      are never activable and all edges incident to $q$ and $r$ are
      stable. Hence $qr$ is stable and we assume without loss
      of generality that $qr$ is directed as $\overrightarrow{qr}$ in
      $C_{i\geq i_0}$.

      Notice that $p$ is a sink in $C$ and that from
      \ref{rule:consecutive} all edges incident to $r$ except
      $\overrightarrow{rp}$ are incoming to $r$ in $C_i$. Edges in
      $C'_{i} \setminus \{p\}$ have the same orientation as in
      $C_{i}$. Using the arguments in Case
      \ref{case:60-stabilize-in-out},
      \ref{rule:all-edges-directed}--\ref{rule:no-cyclic-triangles}
      are satisfied for $q$ and for $r$ in $C'_i$. So, $q$ and
      $r$ are never activable in $C'_i$. For the reasons in
      Case \ref{case:60-stabilize-in-out}, $\mathcal{S}'_{i_0}$ is a
      fair execution of Algorithm \ref{algo:ss-le} on
      $\cP\setminus \{p\}$. Since $\cP$ is of minimum size, there exists
      a step $i_1 \geq i_0$ such that $C_{i_1}'$ is a stable
      configuration where all edges are oriented and that contains a
      unique sink $p''$. Since the only outgoing edge of $r$ in
      $C_{i_0}$ is $\overrightarrow{rp}$, $r=p''$ is the unique sink
      of $C_{i_1}'$. Furthermore, since the edges incident to $p,q,r$
      are stable, $C_{i_1}$ is a final configuration where all edges
      are oriented. By our definition of $i_0$, this implies that
      $i_1 = i_0$. Since $p$ has only incoming edges in
      $C_{i_1} = C_{i_0}$, $p$ is a sink in $C_{i_0}$ and $r$ is not a
      sink in $C_{i_0}$ due to $\overrightarrow{rp}$, so $p$ is the
      unique sink in $C_{i_0}$.
    \end{proof}
    Therefore, for any orientation of the edges incident to $p$ in
    $\cS_{i_0}$, $C_{i_0}$ is a directed final configuration
    containing a unique sink. 
    \end{proof}

    Now assume there exist two
    $120^\circ$ particles connected by a path of $180^\circ$
    particles on the boundary of $\mathcal{P}$.

    \setcounter{case}{0} 
    \begin{lemma} \label{lemma:120-particles-stabilise}
      If $\cP$ contains two $120^\circ$ particles $p,p^*$ connected by
      a path of $180^\circ$ particles on the boundary, all edges are
      stable and there is a unique sink in the final configuration.
    \end{lemma}
    
    \begin{proof}
      Let $q$ and $r$ be the neighbours of $p$ on the boundary and let $s$ be
      the common neighbour of $p, q$ and $r$. By
      Lemma~\ref{lemma:unloved-lemma}, we know that $ps$ is stable in
      $\cS_{i_0}$. 
      We split the proof of the lemma into different cases, depending on
      the orientation of $ps$ in $C_{i\geq i_0}$. 

      \begin{case} \label{case:ps-stable}
        The edge between $p$ and $s$ is $\overrightarrow{ps}$.
      \end{case}

      For each $i \geq i_0$, let $C'_i = C_i \setminus \{p\}$ and
      consider the sequence of configurations
      $\mathcal{S}_{i_0}' = C'_{i_0}, C'_{i_0+1}, \ldots, C'_{i},
      \ldots$ For each $i \geq i_0$, either $C_{i+1} = C_i$ or there
      exists $p_i$ such that $C_{i+1}$ is obtained from $C_i$ by
      activating $p_i$ and thus modifying the orientations of edges
      incident to $p_i$.  From Lemma~\ref{lemma:stable-particles},
      since $\overrightarrow{ps}$ is stable, $p$ is never activable
      and the edges $pq$ and $pr$ are stable.  Consequently,
      $p_i \neq p$ for any $i \geq i_0$. The orientations of $pq$ and $pr$ lead to the following cases.
      
      \begin{subcase} \label{case:120-stabilize-ps-out-out} Particle
        $p$ is incident to $\overrightarrow{pq}$,
        $\overrightarrow{ps}$ and $\overrightarrow{pr}$.
      \end{subcase}

      \begin{proof}
        For each $i \geq i_0$, the edges of $C_i'$ have the same
        orientation as in $C_{i}$. So, a particle
        $p' \neq p$ is activable in $C_i'$ if and only if it is
        activable in $C_i$. The configuration obtained by
        activating $p_i$ in $C_i'$ is precisely $C_{i+1}'$ since the
        edges of $C_i'$ have the same orientation as in
        $C_{i}$. Hence, $\mathcal{S}'_{i_0}$ is a fair
        execution of Algorithm \ref{algo:ss-le} on
        $\cP\setminus \{p\}$. By the minimality of the size of $\cP$,
        there exists $i_1 \geq i_0$ such that $C_{i_1}'$ is a
        final configuration with a unique sink $p''$ where all edges are oriented. Since the edges incident to $p$
        are stable, $C_{i_1}$ is a final configuration where all
        edges are oriented. By our definition of $i_0$, this implies
        that $i_1 = i_0$. Since $p$ has only outgoing edges in
        $C_{i_1} = C_{i_0}$, $p$ is not a sink of $C_{i_0}$ and $p''$
        is the unique sink in $C_{i_0}$.
      \end{proof}
      
      \begin{subcase}\label{case:120-stabilize-ps-out-in}
        Particle $p$ is incident to $\overrightarrow{pq}$,
        $\overrightarrow{ps}$ and $\overrightarrow{rp}$. The case where $p$ is incident to $\overrightarrow{qp}$, $\overrightarrow{ps}$ and $\overrightarrow{pr}$ is completely symmetrical.
      \end{subcase}
      \begin{proof}
        Since $\overrightarrow{rp}$ is stable, from Lemma
        \ref{lemma:stable-particles}, $r$ is not activable in $C$ and
        $rs$ is stable. From \ref{rule:no-cyclic-triangles},
        $\overrightarrow{rs}$ is in $C_{i\geq i_0}$. The edges of
        $C'_{i \geq i_0} \setminus \{p\}$ have the same orientation as
        in $C_{i \geq i_0}$. So for any
        $p' \notin \{p,r\}$, $p'$ is activable in $C_i'$ if and only
        if $p'$ is activable in $C_i$.  Let us consider $r$. From
        Observation \ref{obs:edge-change-dir-consecutive} and the
        incoming edge from the empty common neighbour of $p$ and $r$
        to $r$, \ref{rule:consecutive} is satisfied for $r$ in
        $C'_i$.
        \ref{rule:all-edges-directed}, \ref{rule:incoming-edge} and \ref{rule:no-cyclic-triangles} are satisfied for $r$ in $C'_i$ from
        Observation \ref{obs:r1-r2-r4-hold}. Hence, $r$ is never
        activable in $C'_i$. Furthermore, the configuration obtained
        by activating $p_i \neq r$ in $C_i'$ is precisely $C_{i+1}'$
        since the edges of $C_i'$ have the same orientation as in
        $C_{i}$. So, $\mathcal{S}'_{i_0}$ is a fair execution
        of Algorithm \ref{algo:ss-le} on $\cP\setminus \{p\}$. By the
        minimality of the size of $\cP$, there exists a step
        $i_1 \geq i_0$ such that $C_{i_1}'$ is a final configuration
        where all edges are oriented and that contains a unique sink
        $p'' \neq r$, since $r$ has an outgoing edge
        $\overrightarrow{rs}$ in $C'_i$. Since the edges incident to
        $p$ are stable, $C_{i_1}$ is a stable configuration where all
        edges are oriented. By our definition of $i_0$, this implies
        that $i_1 = i_0$. Since $p$ has outgoing edges in
        $C_{i_1} = C_{i_0}$, $p$ is not a sink of $C_{i_0}$ and $p''$
        is the unique sink in $C_{i_0}$.
      \end{proof}

      \begin{subcase}\label{case:120-stabilize-ps-in-in}
        Particle $p$ is incident to $\overrightarrow{qp}$,
        $\overrightarrow{ps}$ and $\overrightarrow{rp}$.
      \end{subcase}
      \begin{proof}
        Since $\overrightarrow{qp}$ and $\overrightarrow{rp}$ are
        stable, from Lemma \ref{lemma:stable-particles}, $r$ and $q$
        are not activable in $C$ and the edges $rs$ and $qs$ are
        stable. By \ref{rule:no-cyclic-triangles},
        $\overrightarrow{rs}$ and $\overrightarrow{qs}$ belong to
        $C_{i\geq i_0}$. The edges of $C'_{i \geq i_0} \setminus \{p\}$ have the same
        orientation as in $C_{i \geq i_0}$. Consequently, for any
        $p' \notin \{p,q,r\}$, $p'$ is activable in $C_i'$ if and only
        if $p'$ is activable in $C_i$.  Let us consider $q$ and
        $r$. Using the same arguments as in Case
        \ref{case:120-stabilize-ps-out-in},
        \ref{rule:all-edges-directed}--\ref{rule:no-cyclic-triangles}
        are satisfied for $q$ and for $r$ in $C'_i$. Consequently,
         $q$ and $r$ are never activable in $C'_i$.  Furthermore,
        the configuration obtained by activating
        $p_i \notin \{p,q,r\}$ in $C_i'$ is precisely $C_{i+1}'$ since
        the edges of $C_i'$ have the same orientation as in
        $C_{i}$. Therefore, $C'_i$ satisfies all rules and
        $\mathcal{S}'_{i_0}$ is a fair execution of Algorithm
        \ref{algo:ss-le} on $\cP\setminus \{p\}$. By the minimality of
        the size of $\cP$, there exists a step $i_1 \geq i_0$ such
        that $C_{i_1}'$ is a final configuration where all edges are
        oriented and that contains a unique sink $p''$. Since $q$ and
        $r$ have both an outgoing edge to $s$ in $C'_i$,
        $p'' \notin \{q,r\}$.  Since the edges incident to $p$ are
        stable, $C_{i_1}$ is a final configuration where all edges are
        oriented. By our definition of $i_0$, this implies that
        $i_1 = i_0$. Since $p$ has outgoing edges in
        $C_{i_1} = C_{i_0}$, $p$ is not a sink of $C_{i_0}$ and $p''$
        is therefore the unique sink in $C_{i_0}$.
      \end{proof}

      We now consider the case where the edge between $p$ and $s$ is
      $\overrightarrow{sp}$.  Observe that by
      Lemma~\ref{lemma:stable-particles}, $sq$ and $sr$ are stable
      edges.  Note that by Lemma~\ref{lemma:unloved-lemma}, $pq$ is
      stable if and only if $pr$ is stable. We distinguish two
      cases depending on whether these two edges are stable.
      
      \begin{case}\label{case:incident-p-stable}
        The edge between $p$ and $s$ is $\overrightarrow{sp}$, and the
        edges $pq$ and $pr$ are stable.
      \end{case}

      The possible orientations of the stable edges $sq$ and $sr$ in $\mathcal{S}_{i_0}$ give the following subcases.

      \begin{subcase}\label{case:120-stabilize-sp-in-in-rs}
        Particle $s$ is incident to $\overrightarrow{qs}$ and
        $\overrightarrow{rs}$.
      \end{subcase}
      
      \begin{proof}
        By Lemma~\ref{lemma:stable-particles}, $r$ and $q$ are never
        activable and $rp$ and $qp$ are stable. By
        \ref{rule:no-cyclic-triangles}, $\overrightarrow{qp}$ and
        $\overrightarrow{rp}$ are in $C_{i\geq i_0}$.      

        The edges of $C'_{i \geq i_0} \setminus \{p\}$ have the same
        orientation as in $C_{i \geq i_0}$ and thus for any
        $p' \notin \{q,r,s\}$, $p'$ is activable in $C'_i$ if and only
        if it is activable in $C_i$.  Let us consider $q, r, s$. Using
        the same arguments as in Case
        \ref{case:120-stabilize-ps-out-in},
        \ref{rule:all-edges-directed}--\ref{rule:no-cyclic-triangles}
        are satisfied for $q$ and $r$ in $C \setminus \{p\}$.  Since
        $sp$ is between two incoming edges from Observation
        \ref{obs:edge-change-dir-consecutive}, \ref{rule:consecutive}
        is satisfied for $s$ in $C'_i$. \ref{rule:all-edges-directed}, \ref{rule:incoming-edge} and \ref{rule:no-cyclic-triangles} are
        satisfied for $s$ in $C'_i$ from Observation
        \ref{obs:r1-r2-r4-hold}. So, $q,r,s$ are never
        activable in $C'_i$. Furthermore, the configuration obtained
        by activating $p_i \neq q,s,r$ in $C_i'$ is precisely
        $C_{i+1}'$ since the edges of $C_i'$ have the same orientation
        as in $C_{i}$. Hence, $\mathcal{S}'_{i_0}$ is a fair execution
        of Algorithm \ref{algo:ss-le} on $\cP\setminus \{p\}$. By the
        minimality of the size of $\cP$, there exists a step
        $i_1 \geq i_0$ such that $C_{i_1}'$ is a final configuration with a unique sink $p''$ where all edges are oriented. Notice that by \ref{rule:consecutive}, the only
        outgoing edge of $s$ in $C_i$ is $\overrightarrow{sp}$, so
        $p''=s$ is the unique sink of $C'_i$. Since the edges incident
        to $p$ are stable, $C_{i_1}$ is a final configuration where
        all edges are oriented. By our definition of $i_0$, this
        implies that $i_1 = i_0$. Since $s$ has an outgoing edge
        $\overrightarrow{sp}$ in $C_{i_1} = C_{i_0}$, $s$ is not a
        sink of $C_{i_0}$ and $p$ is the unique sink in $C_{i_0}$.
      \end{proof}

      Observe that if $s$ is incident to $\overrightarrow{qs}$ (resp.,
      $\overrightarrow{rs}$), then since $qp$ (resp., $rp$) is stable,
      by \ref{rule:no-cyclic-triangles}, $p$ is incident to
      $\overrightarrow{qp}$ (resp., $\overrightarrow{rp}$). When among
      $qs$ and $rs$, there is one outgoing and one incoming edge, we
      consider two cases depending on whether $p$ is a sink in
      $\cS_{i_0}$.

      \begin{subcase} \label{case:120-stabilize-sp-in-in-sr}
        Particle $s$ is incident to $\overrightarrow{qs}$ and
        $\overrightarrow{sr}$ and $p$ is incident to
        $\overrightarrow{rp}$.          
      \end{subcase}

      \begin{proof}
        Since $\overrightarrow{qs}$ and $\overrightarrow{sp}$ are
        stable, necessarily $\overrightarrow{qp}$ is stable and $p$ is
        a sink in $C_{i \geq i_0}$. So, $p$ is never
        activable in $\cS_{i_0}$ and from
        Lemma~\ref{lemma:stable-particles}, $q,r,s$ are never
        activable in $\cS_{i_0}$ either.
        
        The edges of $C'_{i \geq i_0} \setminus \{p\}$ have the same
        orientation as in $C_{i \geq i_0}$ and thus any particle
        $p' \notin \{p, q,r,s\}$ is activable in $C'_i$ if and only if
        it is activable in $C_i$. Let us consider $q, r,s$. Since
        $\overrightarrow{sp}$ is between an incoming and an outgoing
        edge, by Observation \ref{obs:edge-change-dir-consecutive},
        \ref{rule:consecutive} is satisfied for $s$ in
        $C'_i$ and the remaining rules are
        satisfied for $s$ in $C'_i$ from Observation
        \ref{obs:r1-r2-r4-hold}. The arguments for $q$ are the same as in
        Case \ref{case:120-stabilize-sp-in-in-rs}. Due to the incoming
        edge from the common empty neighbour of $r$ and $p$ and from
        Observation \ref{obs:edge-change-dir-consecutive},
        \ref{rule:consecutive} is satisfied for $r$ in
        $C'_i$ and
        the remaining rules are satisfied for $r$ in $C'_i$ from
        Observation \ref{obs:r1-r2-r4-hold}. So, $q,r,s$ are
        never activable in $C'_{i\geq i_0}$. Furthermore, the
        configuration obtained by activating $p_i \notin \{q,r,s\}$ in
        $C_i'$ is precisely $C_{i+1}'$ since the edges of $C_i'$ have
        the same orientation as in $C_{i}$. Hence, all rules are
        satisfied in $C'_i$ and $\mathcal{S}'_{i_0}$ is a fair
        execution of Algorithm \ref{algo:ss-le} on
        $\cP\setminus \{p\}$. By the minimality of the size of $\cP$,
        there exists a step $i_1 \geq i_0$ such that $C_{i_1}'$ is a
        final configuration where all edges are oriented and that
        contains a unique sink $p''$. Notice that by
        \ref{rule:consecutive}, the only outgoing edge of $r$ in
        $C_i$ is $\overrightarrow{rp}$, so $p''=r$ is the unique sink
        of $C'_i$. As the edges incident to $p$ are stable,
        $C_{i_1}$ is a final configuration where all edges are
        oriented. By our definition of $i_0$,
        $i_1 = i_0$. Since $p$ does not have outgoing edges in
        $C_{i_1} = C_{i_0}$, $r$ is not a sink of $C_{i_0}$ due to
        $\overrightarrow{rp}$ and $p$ is the unique sink in $C_{i_0}$.
      \end{proof}

      \begin{subcase}\label{case:120-stabilize-sp-in-out}
        Particle $s$ is incident to $\overrightarrow{qs}$ and
        $\overrightarrow{sr}$ and $p$ is incident to
        $\overrightarrow{pr}$.          
      \end{subcase}
      
      \begin{proof}
        As noted before the stable edge $qp$ is oriented as
        $\overrightarrow{qp}$ by \ref{rule:no-cyclic-triangles}. By
        Lemma~\ref{lemma:stable-particles}, $p$, $q$ and $s$ are never
        activable in $\cS_{i_0}$ and the edges $pr$ and $sr$ are
        stable.  The edges of $C'_{i \geq i_0} \setminus \{p\}$ have
        the same orientation as in $C_{i \geq i_0}$. So,
        for any $p' \notin \{p,q,s\}$, $p'$ is activable in $C'_i$ if
        and only if it is activable in $C_i$. Let us consider $q$ and
        $s$. Using the same arguments for both $q$ and $s$ as in Case
        \ref{case:120-stabilize-sp-in-in-sr}, we obtain that
        \ref{rule:all-edges-directed}--\ref{rule:no-cyclic-triangles}
        are always satisfied at $q$ and $s$, and that they are never
        activable in $C_{i\geq i_0}'$. Hence, $\mathcal{S}'_{i_0}$ is
        a fair execution of Algorithm \ref{algo:ss-le} on
        $\cP\setminus \{p\}$. By the minimality of the size of $\cP$,
        there exists a step $i_1 \geq i_0$ such that $C_{i_1}'$ is a
        final configuration where all edges are oriented and that
        contains a unique sink $p'' \notin \{q,s\}$, since $q$ and $s$
        have outgoing edges in $C \setminus \{p\}$. Since the edges
        incident to $p$ are stable, $C_{i_1}$ is a stable
        configuration where all edges are oriented. By our definition
        of $i_0$, this implies that $i_1 = i_0$. Since $p$ is incident
        to $\overrightarrow{pr}$ in $C_{i_1} = C_{i_0}$, $p$ is not a
        sink in $C_{i}$, hence $p''$ is the unique sink in $C_{i_0}$.
      \end{proof}

      \begin{subcase}\label{case:120-stabilize-s-out}
        Particle $s$ is incident to $\overrightarrow{sq}$ and
        $\overrightarrow{sr}$.
      \end{subcase}

      \begin{proof}
        Since $pq, pr, ps$ are stable and since
        $ps$ is directed as $\overrightarrow{sp}$ in $C_{i\geq i_0}$, $p$ is incident to at
        most one outgoing edge. Without loss of generality, we can
        thus assume that $\overrightarrow{qp}$ is in $C_{i\geq i_0}$.
        Observe that by \ref{rule:consecutive},
        $\overrightarrow{qp}$ is the only outgoing edge at $q$ and no
        neighbour of $q$ is activable in $C_{i\geq i_0}$ by
        Lemma~\ref{lemma:stable-particles}. Note also that $s$ has
        three outgoing edges
        $\overrightarrow{sq}, \overrightarrow{sp},
        \overrightarrow{sr}$ in $C_i$ and since $s$ is not activable
        in $C_i$, by \ref{rule:incoming-edge}, all other edges
        incident to $s$ are incoming. Again, this implies that all
        neighbours of $s$ different from $r$ are not activable in
        $C_{i\geq i_0}$. For each $i \geq i_0$, let $C^*_i$ be the
        configuration obtained from $C_i$ by replacing
        $\overrightarrow{sq}$ by $\overrightarrow{qs}$.

        For any particle $p' \notin \{s,q\}$,
        \ref{rule:incoming-edge} and \ref{rule:consecutive} are
        satisfied at $p'$ in $C^*_{i}$ since they are satisfied at
        $p'$ in $C_i$ by Lemma~\ref{lemma:hated-lemma} and the
        orientation of the edges incident to $p'$ in $C_i^*$ is the
        same as in $C_i$. 
        Since $q$ only has one outgoing edge in
        $C_i$, \ref{rule:incoming-edge} and \ref{rule:consecutive} are satisfied at $q$ in $C^*_{i}$. Since $s$ has
        three outgoing edges
        $\overrightarrow{sq}, \overrightarrow{sp},
        \overrightarrow{sr}$ in $C_i$ and since all other edges
        incident to $s$ are incoming, \ref{rule:incoming-edge} directly holds at $s$ in $C_i^*$ and \ref{rule:consecutive} holds at $s$ in
        $C_i^*$ since $p,q$ and $r$ are reached through consecutive ports of $s$ by definition. If
        \ref{rule:no-cyclic-triangles} is not satisfied at some
        particle $p'$ in $C_i^*$, there is a directed triangle
        made of the edges
        $\overrightarrow{qs}, \overrightarrow{sp'},
        \overrightarrow{p'q}$ in $C_i^*$.  Since the only
        out-neighbours of $s$ in $C_i^*$ are $p$ and $r$, necessarily,
        $p' =p$, but this is impossible since $pq$ is oriented from
        $q$ to $p$. So, \ref{rule:incoming-edge},
        \ref{rule:consecutive}, \ref{rule:no-cyclic-triangles} are
        always satisfied in $C^*_{i \geq i_0}$. Since all edges
        incident to $q$ and $s$ are stable in $C_i$,
        \ref{rule:all-edges-directed} is also satisfied at $s$ and
        $q$ in $C_i$ and in $C_i^*$. Hence, $q$ and $s$ are
        never activable in $C^*_{i \geq i_0}$. For any $p' \notin \{q,s\}$, $p'$
        is activable in $C_i^*$ if and only if it is activable in
        $C_i$. Therefore, $\cS_i^*$ is a fair execution of
        Algorithm~\ref{algo:ss-le} on $\cP$. Note that when
        considering $C^*_{i\geq i_0}$, we are in Case~\ref
        {case:120-stabilize-sp-in-in-sr}
        or~\ref{case:120-stabilize-sp-in-out}. So, we know
        that $C^*_{i_0}$ is a final configuration where all edges are
        oriented and that contains a unique sink $p''$ different from
        $q$ and $s$. Since a particle $p'$ is activable in $C_{i_0}$
        if and only if it is activable in $C^*_{i_0}$, $C_{i_0}$ is
        also a valid configuration, and $p''$ is the unique sink of
        $C_{i_0}$.
      \end{proof}

      Finally, we consider the case where the edges $pq$ and $pr$ are
      not stable. We remind the reader that from Lemma \ref{lemma:unloved-lemma}, $pq$ and $pr$ are either both stable or both unstable. 
      
      \begin{case} \label{case:incident-p-not-stable} The edge between
        $p$ and $s$ is $\overrightarrow{sp}$ and the edges $pq$ and
        $pr$ are not stable.
      \end{case}

      \begin{proof}
        We will prove that this case is not possible.
        \begin{figure}[ht]
          \centering
          \begin{subfigure}{.49\textwidth}
            \centering
            \begin{tikzpicture}[scale=.4]
              \clip (0,1.73) rectangle (16,5*1.4) ;
              
              \filldraw[gray!10] (0:0) ++ (-120:2) --++ (90:3*1.73)
              --++ (0:18) --++ (-90:3*1.73) ;
              
              \triangularGrid{2}{8}{2}{1.73}
              
              \filldraw (60:4) circle (6pt) ++ (60:2) node[]
              {\scriptsize $\blacksquare$} ++ (-60:2) circle (6pt) ++
              (60:2) circle (6pt) ++ (-60:2) circle (6pt) ++ (60:2)
              circle (6pt) ++ (-60:2) circle (6pt) ++ (60:1) ++ (0:1)
              node[] {$\ldots$} ++ (60:1) ++ (0:1) circle (6pt) ++
              (-60:2) circle (6pt) ++ (60:2) circle(6pt) ++ (-60:2)
              circle (6pt) ;
              
              \draw[-] (60:4) --++ (60:1.7) ;
              \draw[-] (60:6) --++ (0:1.7) ;
              \draw[-Stealth] (0:2) ++ (60:4) --++ (180:1.7) ;
              \draw[-Stealth] (0:2) ++ (60:4) --++ (120:1.7) ;
              \draw[-Stealth] (0:2) ++ (60:4) --++ (60:1.7) ;
              
              \filldraw (0:.5) ++ (60:3.1) node[] {\small $q$} ++
              (0:2) node[] {\small $s$} ++ (0:2) node[] {\small $s_2$}
              ++ (0:2) node[] {\small $s_3$} ++ (0:4) node[] {\small
                $s_{k}$} ++ (0:2) node[] {\small $s_{k+1}$} ++
              (90:3.2) ++ (180:1) node[] {\small $r_k$} ++ (180:2)
              node[] {\small $r_{k-1}$} ++ (180:4) node[] {\small
                $r_2$} ++ (180:2) node[] {\small $r$} ++ (180:2)
              node[] {\small $p$} ;
            \end{tikzpicture}
            \caption{A
            $120^\circ$ particle $p$ (square) with
            $\{pq,pr\}$ unstable and $s$ incident to
            the directed edges $sq$, $sp$ and
            $sr$}
            \label{fig:120-sp-undirected-original}
          \end{subfigure}
          \begin{subfigure}{.49\textwidth}
            \centering
            \begin{tikzpicture}[scale=.4]
              \clip (0,1.73) rectangle (16,5*1.4) ;
              
              \filldraw[gray!10] (0:0) ++ (-120:2) --++ (90:3*1.73)
              --++ (0:18) --++ (-90:3*1.73) ;
              
              \triangularGrid{2}{8}{2}{1.73}
              
              \filldraw (60:4) circle (6pt) ++ (60:2) node[]
              {\scriptsize $\blacksquare$} ++ (-60:2) circle (6pt) ++
              (60:2) circle (6pt) ++ (-60:2) circle (6pt) ++ (60:2)
              circle (6pt) ++ (-60:2) circle (6pt) ++ (60:1) ++ (0:1)
              node[] {$\ldots$} ++ (60:1) ++ (0:1) circle (6pt) ++
              (-60:2) circle (6pt) ++ (60:2) circle(6pt) ++ (-60:2)
              circle (6pt) ;
              
              \draw[-] (60:4) --++ (60:1.7) ;
              \draw[-] (60:6) --++ (0:1.7) ;
              \draw[-Stealth] (0:2) ++ (60:4) --++ (180:1.7) ;
              \draw[-Stealth] (0:2) ++ (60:4) --++ (120:1.7) ;
              \draw[-Stealth] (0:2) ++ (60:4) --++ (60:1.7) ;
              
              \draw[-] (60:6) ++ (0:2) --++ (0:2) ++ (0:4) --++ (0:2) ;
              \draw[-Stealth] (0:16) ++ (120:4) --++ (120:1.7) ;
              \draw[-Stealth] (0:16) ++ (120:4) --++ (180:1.7) ;
              \draw[-Stealth] (0:16) ++ (120:4) ++ (180:2) --++ (60:1.7) ;
              \draw[-Stealth] (0:16) ++ (120:4) ++ (180:2) --++ (120:1.7) ;
              \draw[-Stealth] (0:16) ++ (120:4) ++ (180:2) --++ (180:1.7) ;
              \draw[-Stealth] (0:10) ++ (120:4) --++ (180:1.7)  ;
              \draw[-Stealth] (0:8) ++ (120:4) --++ (180:1.7)  ;
              \draw[-Stealth] (0:4) ++ (60:4) --++ (120:1.7) ;
              \draw[-Stealth] (0:4) ++ (60:4) --++ (60:1.7) ;
              \draw[-Stealth] (0:6) ++ (60:4) --++ (120:1.7) ;
              
              \filldraw (0:.5) ++ (60:3.1) node[] {\small $q$} ++
              (0:2) node[] {\small $s$} ++ (0:2) node[] {\small $s_2$}
              ++ (0:2) node[] {\small $s_3$} ++ (0:4) node[] {\small
                $s_{k}$} ++ (0:2) node[] {\small $s_{k+1}$} ++
              (90:3.2) ++ (180:1) node[] {\small $r_k$} ++ (180:2)
              node[] {\small $r_{k-1}$} ++ (180:4) node[] {\small
                $r_2$} ++ (180:2) node[] {\small $r$} ++ (180:2)
              node[] {\small $p$} ;
            \end{tikzpicture}
            \caption{The final orientation of edges
            in Case \ref{case:incident-p-not-stable}\newline\newline}
            \label{fig:120-sp-undirected-final}
          \end{subfigure}
          \caption{The setting in Case \ref{case:incident-p-not-stable}}
        \end{figure}

        Without loss of generality, assume that $r$ is on the path
        connecting $p$ to $p^*$ via $180^\circ$ particles. Note that it
        is possible that $p^* = r$. Let
        $(r_0= p, r_1=r, r_2, \ldots, r_k = p^*)$ be the path on the
        boundary from $p$ to $p^*$ whose inner particles are all
        $180^\circ$ particles (if $r =p^*$, then $k=1$). Let $s_{j+1}$
        be the common neighbour of any pair of consecutive
        particles $r_j$ and $r_{j+1}$ with $0 \leq j \leq k-1$, and
        observe that $s_1 = s$. Let $s_0= q$ and let $s_{k+1}$ be the
        neighbour of $r_k$ on the boundary that is distinct from
        $r_{k-1}$. This setting is also shown in Figure
        \ref{fig:120-sp-undirected-original}.

        Since $s$ is incident to a stable edge, $\overrightarrow{sp}$, from Lemma~\ref{lemma:stable-particles} all edges incident to
        $s$ are stable. By Lemma~\ref{lemma:stable-particles} applied
        at $q$ and $r$ and since $pq$ and $pr$ are not stable,
        necessarily $\overrightarrow{sq}$ and $\overrightarrow{sr}$
        are stable in $\cS_{i_0}$.  Since in this setting $s$ has
        three outgoing edges, $sq$, $sp$ and $sr$, from
        \ref{rule:incoming-edge} $ss_2$ is incoming to $s$. So,
        $s_2$ is incident to the stable outgoing edge
        $\overrightarrow{s_2s}$. From Lemma
        \ref{lemma:stable-particles}, $s_2$ is not activable and all
        edges incident to $s_2$ are stable and directed. From
        \ref{rule:no-cyclic-triangles}, $s_2r$ is oriented from $s_2$
        to $r$. If $r = p^*$ (i.e., if $k = 1$), $r$ is incident
        to only one edge that is not stable, which is impossible from
        Lemma~\ref{lemma:unloved-lemma}. Hence, $k \geq 2$ and
        $rr_2$ is not stable.
        As $s_2$ is not activable, from Lemma
        \ref{lemma:stable-particles}, $s_2r_2$ is stable. If $s_2r_2$
        is directed from $r_2$ to $s_2$, $r_2$ has a stable outgoing
        edge and from Lemma \ref{lemma:stable-particles}, $rr_2$ is
        stable, which is a contradiction for $r$. So $s_2r_2$ is
        directed from $s_2$ to $r_2$. Generalising, for $i \geq 1$ each particle $s_i$ is incident to the stable edges $\overrightarrow{s_is_{i-1}}$, $\overrightarrow{s_ir_{i-1}}$ and $\overrightarrow{s_ir_i}$ and the edge $r_{i-1}r_i$ is not stable.
        Then, for $i=k$ the edge $r_{k-1}r_{k}$ should be the only unstable edge incident to $r_k$ which is impossible from Lemma \ref{lemma:unloved-lemma}, a contradiction.  
      \end{proof}

      This ends the proof of
      Lemma~\ref{lemma:120-particles-stabilise}.
    \end{proof}

    The proof of Theorem~\ref{th:algo-correct} follows from
    Lemmas~\ref{lemma:boundaryparticles}, \ref{lemma:pending},
    \ref{lemma:60-particles-stabilise} and
    \ref{lemma:120-particles-stabilise}.
\end{proof}

\section{Further Remarks}

A very nice property of our algorithm is that it is very
simple. However, we showed that our algorithm works assuming that the
scheduler is Gouda fair. 
The execution presented in Figure~\ref{fig:problematic-execution}
shows that if we consider a sequential unfair scheduler (i.e., we only
ask that the scheduler activates an activable particle at each step),
there exist periodic executions that never reach a valid
configuration. It would thus be interesting to understand if we can
design a self-stabilising leader election algorithm for simply
connected configurations that is correct even with an unfair
scheduler.
%In the case where this is possible, we believe that this would lead to
%a much more complex algorithm than our algorithm.

\begin{figure}[htp]
\centering
\includegraphics[scale=0.35,page=2]{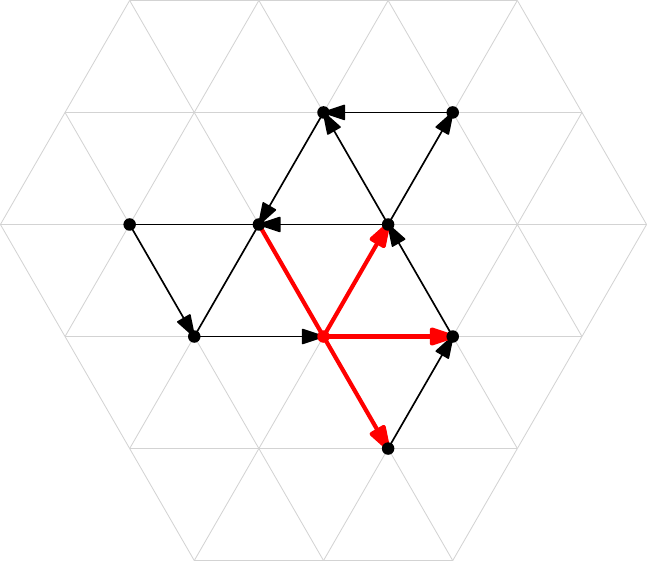}
\includegraphics[scale=0.35,page=5]{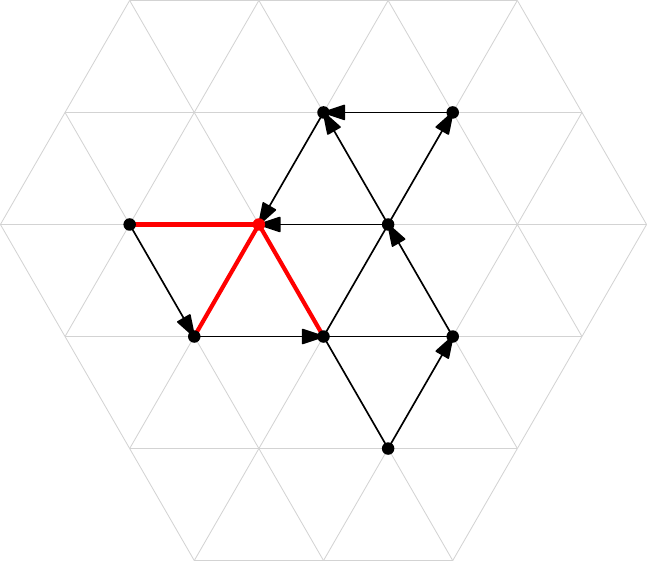} 
\includegraphics[scale=0.35,page=8]{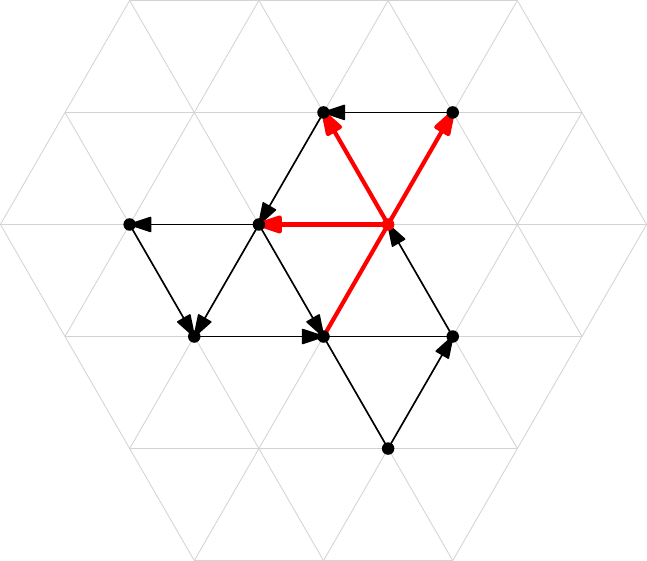}

\

\includegraphics[scale=0.35,page=11]{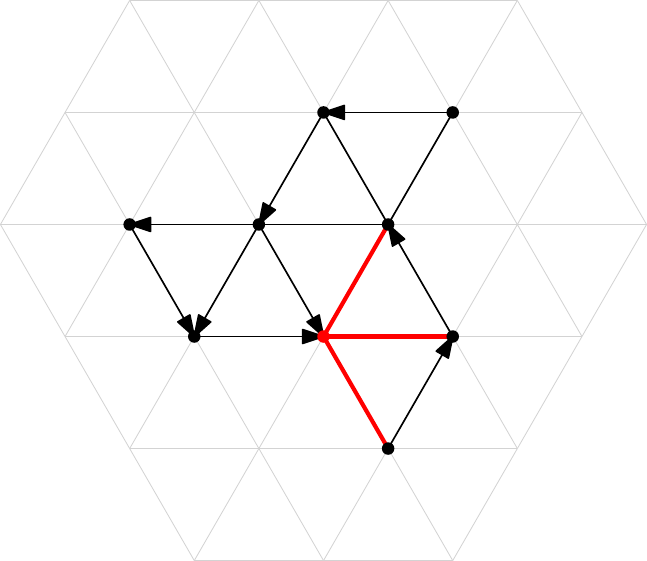} 
\includegraphics[scale=0.35,page=14]{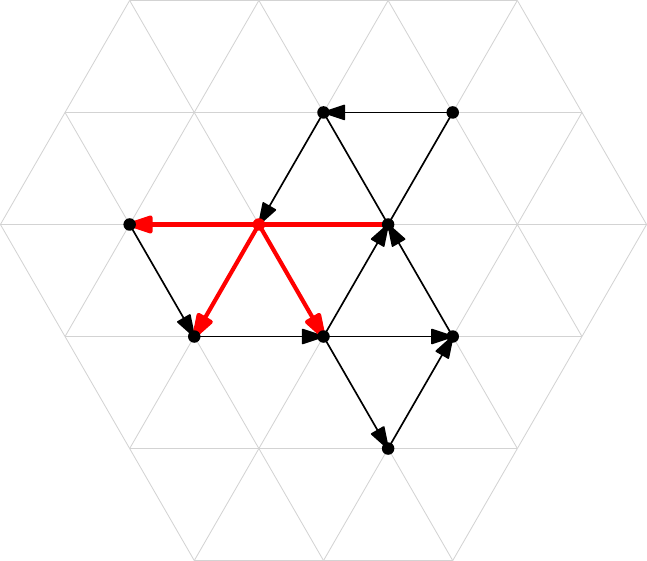}
\includegraphics[scale=0.35,page=17]{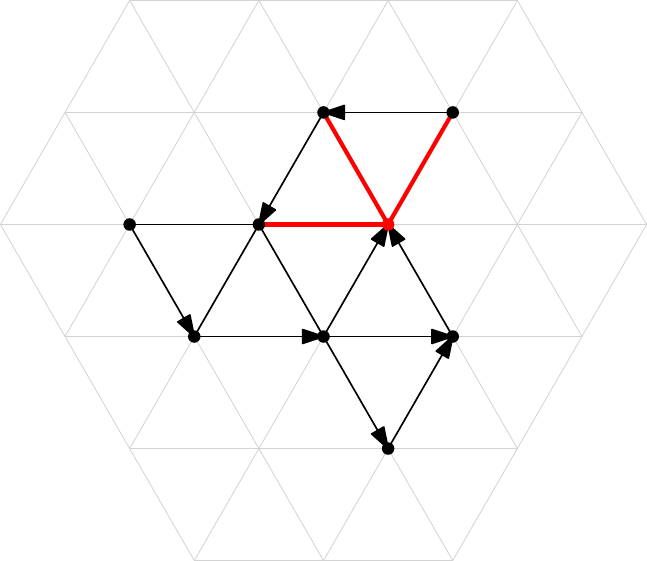}
\caption{A periodic unfair execution of our algorithm. At each step,
  the red vertex is activated, and it modifies the status of its
  incident red edges (i.e., the ones that are not incoming).}
\label{fig:problematic-execution}
\end{figure}

Our algorithm heavily uses the geometry of the system and relies on
the fact that the support is simply connected. If we suppose that the
particles agree on the orientation of the grid, the impossibility
results of Dolev et al.~\cite{DolevGS99} no longer hold. One can
thus wonder if, assuming that the particles agree on the orientation of
the grid, it is possible to design a silent self-stabilising leader
election algorithm using constant memory for arbitrary connected
configurations. Again, one should use the geometry of the grid in
order to overcome the impossibility results of Dolev et al., but it
seems to be very challenging.

Another natural direction is to consider different classes of graphs where a self--stabilising leader election algorithm with constant memory can be obtained. For example, we can use similar arguments to the ones in this work to show that Algorithm \ref{algo:ss-le} also elects a unique leader in the class of planar graphs in which all inner faces are triangles and the degree of each inner node is at least six. A closely related open question is to consider whether it is possible to obtain self--stabilising leader election algorithms with constant memory for some class of graphs, under weaker assumptions on the scheduler (i.e., without assuming Gouda fairness). 

\bibliographystyle{unsrt}
\bibliography{biblio}

\end{document}